\pdfoutput=1
\RequirePackage{ifpdf}
\ifpdf % We~are running pdfTeX in pdf mode
\documentclass[pdftex]{sigma}
\else
\documentclass{sigma}
\fi

\usepackage{bm, mathrsfs, mathtools}
\usepackage{bbold}

\numberwithin{equation}{section}

\newtheorem{Theorem}{Theorem}[section]

\newtheorem{Lemma}[Theorem]{Lemma}
\newtheorem{Proposition}[Theorem]{Proposition}
 { \theoremstyle{definition}

 }

\DeclareMathOperator{\Res}{Res}
\DeclareMathOperator{\Ad}{Ad}
\DeclareMathOperator{\ad}{ad}

\DeclareMathOperator{\Tr}{Tr}

\newcommand{\lda}{\lambda}

\newcommand{\g}{\mathfrak{g}}
\newcommand{\Lg}{\widetilde{\mathfrak{g}}}

\newcommand{\LG}{\widetilde{G}}

\newcommand{\CP}{\mathbb{P}^1}

\newcommand{\Lag}{\mathscr{L}}

\newcommand{\sgm}{\sigma}
\newcommand{\omg}{\omega}

\newcommand{\ti}[1]{_{\mathbf{\underline{#1}}}}

\newcommand{\longhookrightarrow}{\lhook\joinrel\relbar\joinrel\rightarrow}

\def\CC{\mathbb{C}}
\def\CP{\mathbb{C}P^1}

\def\RR{\mathbb{R}}

\def\ZZ{\mathbb{Z}}

\def\O{\mathcal{O}}
\def\P{\mathcal{P}}

\def\1{\bm{1}}

\begin{document}

%\allowdisplaybreaks

\newcommand{\arXivNumber}{2405.12837}

\renewcommand{\PaperNumber}{100}

\FirstPageHeading

\ShortArticleName{Lagrangian Multiform for Cyclotomic Gaudin Models}

\ArticleName{Lagrangian Multiform for Cyclotomic Gaudin Models}

\Author{Vincent CAUDRELIER~$^{\rm a}$, Anup Anand SINGH~$^{\rm a}$ and Beno\^{\i}t VICEDO~$^{\rm b}$}

\AuthorNameForHeading{V.~Caudrelier, A.A.~Singh and B.~Vicedo}

\Address{$^{\rm a)}$~School of Mathematics, University of Leeds, Leeds LS2 9JT, UK}
\EmailD{\href{mailto:v.caudrelier@leeds.ac.uk}{v.caudrelier@leeds.ac.uk}, \href{mailto:anupanandsingh@gmail.com}{anupanandsingh@gmail.com}}

\Address{$^{\rm b)}$~Department of Mathematics, University of York, York YO10 5DD, UK}
\EmailD{\href{mailto:b.vicedo@york.ac.uk}{b.vicedo@york.ac.uk}}

\ArticleDates{Received May 28, 2024, in final form November 07, 2024; Published online November 15, 2024}

\Abstract{We construct a Lagrangian multiform for the class of cyclotomic (rational) Gaudin models by formulating its hierarchy within the Lie dialgebra framework of Semenov-Tian-Shansky and by using the framework of Lagrangian multiforms on coadjoint orbits. This provides the first example of a Lagrangian multiform for an integrable hierarchy whose classical $r$-matrix is non-skew-symmetric and spectral parameter-dependent. As an important by-product of the construction, we obtain a Lagrangian multiform for the {\it periodic} Toda chain by choosing an appropriate realisation of the cyclotomic Gaudin Lax matrix. This fills a gap in the landscape of Toda models as only the open and infinite chains had been previously cast into the Lagrangian multiform framework. A slightly different choice of realisation produces the so-called discrete self-trapping (DST) model. We demonstrate the versatility of the framework by {\it coupling} the periodic Toda chain with the DST model and by obtaining a Lagrangian multiform for the corresponding integrable hierarchy.}

\Keywords{Lagrangian multiforms; integrable systems; classical $r$-matrix; Gaudin models}

\Classification{17B80; 37J35; 70H06}

\section{Introduction}
A characteristic property of an integrable system is that its equations of motion can be seen as members of a hierarchy of compatible equations. Traditional Lagrangians fail to capture this notion of commuting (Hamiltonian) flows, or its discrete analog known as multidimensional consistency \cite{BS, N1}. This obstacle to describing integrable hierarchies variationally was overcome for the first time in \cite{LN} by introducing a new object -- a Lagrangian multiform -- together with a generalised variational principle applied to an appropriate generalisation of a classical action, and noticing a fundamental property of the multiform: the so-called closure relation. Since its introduction, this idea has been developed in several directions. Its connections with more traditional features of integrability (Lax pairs and Hamiltonian structures, for instance) have been established in the various realms of integrable systems: discrete and continuous finite-dimensional systems \cite{CDS,PS,Su,YKLN}, continuous infinite-dimensional systems -- field theories in $1+1$ dimensions \cite{CS1,CS2,CSV,PV,SNC1,SNC2,SuV} and in $2+1$ dimensions \cite{N2, SNC3} -- and semi-discrete systems \cite{SV}. The relations between discrete and continuous multiforms were explored in \cite{V}. The concept has even been extended to non-commuting flows in \cite{CNSV}.

In general, hierarchies of models in $d$ spacetime dimensions are described by a Lagrangian multiform which is a $d$-form integrated over a hypersurface of dimension $d$ in a so-called multi-time space of dimension greater than $d$ to yield an action functional depending not only on the field configurations but also on the hypersurface itself. This last point is the crucial new ingredient of the generalised variational principle used in Lagrangian multiform theory. One postulates a principle of least action which must be valid for {\it any} hypersurface embedded in the multi-time space. This captures the commutativity of the flows variationally and was adopted as a definition of pluri-Lagrangian systems, see \cite{PS,PV} and references therein. In Lagrangian multiform theory, there is an additional postulate which leads to the {\it closure relation}. The latter has been shown \cite{CDS, Su} to be equivalent to the Poisson involutivity of Hamiltonians, the Liouville criterion for integrability.

The generalised variational principle produces equations that come in two flavours: 1) Euler--Lagrange equations associated with each of the coefficients of the Lagrangian multiform which form a collection of Lagrangian densities; 2) Corner or structure equations on the Lagrangian coefficients themselves which select possible models and ensure the compatibility of the various equations of motion imposed on a common set of fields. This paper deals with a finite-dimensional system, so let us illustrate the ingredients for this case. The basic objects are a~Lagrangian 1-form
\begin{equation*}
%\label{def_Lag_multiform}
 \Lag[q]=\sum_{k=1}^N \Lag_k[q] \, \mathrm{d}t_k
\end{equation*}
and the related generalised action
\begin{equation*}
 S[q,\Gamma]=\int_\Gamma \Lag[q],
\end{equation*}
where $\Gamma$ is a curve in the multi-time $\RR^N$ with (time) coordinates $t_1,\dots,t_N$, and $q$ denotes generic configuration coordinates. For instance, $q$ could be a position vector in $\RR^d$ for some $d$, or as will be the case for us, an element of a (matrix) Lie group. The notations $\Lag[q]$ and $\Lag_k[q]$ mean that these quantities depend on $q$ and a finite number of derivatives of $q$ with respect to the times $t_1,\dots,t_N$. In this paper, we restrict ourselves only to the case of first derivatives and simply write $\Lag_k$ for the Lagrangian coefficients.
The application of the generalised variational principle leads to the following multi-time Euler--Lagrange equations \cite{Su}:
\begin{align}
 \label{simple_multitime_EL1}
 &\frac{\partial \Lag_k}{\partial q}-\partial_{t_k} \frac{\partial \Lag_k}{\partial q_{t_k}}=0 ,\\[1ex]
 \label{simple_multitime_EL2}
 &\frac{\partial \Lag_k}{\partial q_{t_\ell}}=0 ,\qquad \ell\neq k ,\\[1ex]
 \label{simple_multitime_EL3}
 &\frac{\partial \Lag_k}{\partial q_{t_k}}=\frac{\partial \Lag_\ell}{\partial q_{t_\ell}} ,\qquad k,\ell=1,\dots,N .
\end{align}
Note that \eqref{simple_multitime_EL1} is simply the standard Euler--Lagrange equation for each $\Lag_k$. Equations \eqref{simple_multitime_EL2} and \eqref{simple_multitime_EL3} are the corner equations. It turns out that these corner equations will, in fact, be identically satisfied by our Lagrangian multiforms, but in general they can represent non-trivial equations for $q$.
The closure relation then stipulates that
\begin{equation*}
%\label{closure}
 \mathrm{d}\Lag[q]=0\quad\Leftrightarrow \quad\partial_{t_k}\Lag_j-\partial_{t_j}\Lag_k=0
\end{equation*}
on solutions of \eqref{simple_multitime_EL1}--\eqref{simple_multitime_EL3}.

Gaudin models are a general class of integrable systems associated with Lie algebras with a~nondegenerate invariant bilinear form. They were first introduced for the Lie algebra $\mathfrak{sl}(2, \mathbb{C})$ by M. Gaudin in \cite{G1} as quantum integrable spin chains with long-range interactions, and then generalised for arbitrary semi-simple Lie algebras in \cite{G2}.
At both the classical and quantum levels, the integrable structure of this model, associated with a Lie algebra $\g$, is underpinned by the rational skew-symmetric solution
\begin{equation}
\label{rational_r_mat}
 r_{\ti{12}}^0(\lda, \mu) = \frac{C_{\ti{12}}}{\mu - \lda}
\end{equation}
of the classical Yang--Baxter equation\footnote{Note that this is the general version of the CYBE applicable to non-skew-symmetric matrices.} (CYBE)
\begin{equation*}
 \bigl[r_{\ti{12}}(\lda, \mu),r_{\ti{13}}(\lda, \nu)\bigr]+\bigl[r_{\ti{12}}(\lda, \mu),r_{\ti{23}}( \mu,\nu)\bigr]+\bigl[r_{\ti{32}}(\nu, \mu),r_{\ti{13}}(\lda, \nu)\bigr]=0 .
\end{equation*}
The rational classical $r$-matrix \eqref{rational_r_mat} takes values in $\g \otimes \g$ and depends on spectral parameters $\lda, \mu \in \mathbb{C}$, and the tensor Casimir $C_{\ti{12}} \in \g \otimes \g$ appearing in \eqref{rational_r_mat} is defined as
\begin{equation*}
 C_{\ti{12}} = I_a \otimes I^a,
\end{equation*}
where $\{I_a\}$ and $\{I^a\}$ are dual bases of $\g$ with respect to a fixed nondegenerate invariant bilinear form. The main property of $C_{\ti{12}}$ which gives it its name is $\bigl[C_{\ti{12}} , X_{\ti 1}+X_{\ti 2}\bigr]=0$ for all $X\in\g$ and where we use the standard tensorial notation $X_{\ti 1} \equiv X \otimes \mathbb{1}$ and $X_{\ti 2} \equiv \mathbb{1} \otimes X$.

The cyclotomic Gaudin model of interest in this paper arises as a specialisation of a general procedure which can be traced back to reduction group ideas \cite{M}, first applied in the form of an averaging procedure \cite{RF} to the rational classical $r$-matrix \eqref{rational_r_mat} in order to produce the trigonometric and elliptic $r$-matrices. This was generalised in various ways, for instance, in~\mbox{\cite{Av1,Av2, AT}} in~the~context of Sklyanin's (linear) Poisson algebra
\begin{equation}
\label{Sklyanin_algebra}
 \{U_{\ti 1}(\lda) ,U_{\ti 2}(\mu)\}=\bigl[ r_{\ti{12}}^0(\lda, \mu),U_{\ti 1}(\lda)+U_{\ti 2}(\mu)\bigr],
\end{equation}
where, as above, the index denotes which factor in the tensor product $\g \otimes \g$ the $\g$-valued function~$U$ sits in, that is, $U_{\ti 1}(\lda) \equiv U(\lda) \otimes \mathbb{1}$ and $U_{\ti 2}(\mu) \equiv \mathbb{1} \otimes U(\mu)$.

The idea is to use a Lie algebra automorphism appropriately extended to a loop algebra automorphism. For our purpose, we use an automorphism $\sigma$ of order $T$ on $\g$ and define%
 \begin{equation*}
 \phi \colon \ U(\lda)\mapsto \omega \sigma^{-1}(U(\omega \lda)),
 \end{equation*}
 where $\omega$ is a $T$-th root of unity. Thus, $\phi$ is an automorphism of the Poisson algebra \eqref{Sklyanin_algebra}. This leads us to consider the fixed point subalgebra generated by
 \begin{equation}
 \label{def_averaged_L}
 L(\lda)=\frac{1}{T} \sum_{k=0}^{T-1} \omega^{-k}\sgm^k U\bigl(\omega^{-k}\lda\bigr) , \qquad \sigma(L(\lda))=\omega L(\omega \lda) .
 \end{equation}
The Poisson algebra for $L(\lda)$ closes into
 \begin{equation}
 \label{PB_L}
 \bigl\{L_{\ti 1}(\lda) ,L_{\ti 2}(\mu)\bigr\}=\bigl[ r_{\ti{12}}(\lda, \mu) ,L_{\ti 1}(\lda)\bigr]-\bigl[ r_{\ti{21}}(\mu,\lda) ,L_{\ti 2}(\mu)\bigr],
\end{equation}
where
\begin{equation}
\label{eq:nss-rmat}
 r_{\ti{12}}(\lda, \mu) = \frac{1}{T} \sum_{k=0}^{T-1} \frac{\sgm_{\ti{1}}^k C_{\ti{12}}}{\mu - \omg^{-k}\lda}
\end{equation}
is a non-skew-symmetric solution of CYBE. These facts are special cases of \cite[Propositions~2.2 and~4.1]{CC}. It is well known \cite{BV} that the Poisson algebra~\eqref{PB_L} ensures that the quantities ${\rm Tr} L(\lda)^p$ Poisson commute and the equations of motion generated by them with respect to \eqref{PB_L} take the Lax form. More precisely, the generating Hamiltonian
\begin{equation*}
 H(\mu)=\frac{1}{p+1} \Tr L(\mu)^{p+1} , \qquad p\in \ZZ_{\ge 0} ,
\end{equation*}
satisfies
\begin{equation*}
 \{L(\lda) ,H(\mu)\} = [M(\lda,\mu) ,L(\lda)],
\end{equation*}
where by definition $\{L(\lda) ,H(\mu)\}=\{L^a(\lda) ,H(\mu)\}I_a$, and
\begin{equation*}
 M_{\ti 1}(\lda,\mu)= \Tr_{\ti 2}r_{\ti {12}}(\lda,\mu)L_{\ti 2}(\mu)^p .
\end{equation*}
This yields an infinite number of Hamiltonians by specifying $p$ and extracting coefficients at the poles of $L(\mu)$. Of course, only a finite number of these Hamiltonians are independent when acting on a finite-dimensional phase space. They generate mutually compatible flows forming an integrable hierarchy.

If one applies this construction to a Lax matrix of rational Gaudin model type, with poles in the finite set $D = \{0, \zeta_1, \dots, \zeta_N, \infty\} \subset \mathbb{C}P^1$,
\begin{equation}
\label{Lax_matrix_Gaudin}
 U(\lda)=\sum_{n=0}^{N_0 - 1} \frac{U_0^{(n)}}{\lda^{n+1}} + \sum_{r=1}^{N} \sum_{n=0}^{N_{r}-1} \frac{U_{r}^{(n)}}{(\lda - \zeta_r)^{n+1}} + \sum_{n=0}^{N_\infty} U_\infty^{(n)} \lda^n , \qquad U_{s}^{(n)}\in \g ,
\end{equation}
then one obtains the so-called cyclotomic Gaudin model \cite{S, VY1, VY2} and its associated integrable hierarchy. Note that in \eqref{Lax_matrix_Gaudin}, we have anticipated that the pole at $0$ and at infinity behave differently in \eqref{def_averaged_L} compared to the poles at $\lda=\zeta_r\neq 0,\infty$. Although we will present the general algebraic setup, for simplicity, in our examples we will only consider the case where~$U(\lda)$ has simple poles at all $\zeta_r$, $r \in \{1, \dots, N\}$, and double poles at the origin and at
infinity. This will be sufficient for our application to the periodic Toda chain and the discrete self-trapping (DST) model. The corresponding Lax matrix of the cyclotomic Gaudin model then takes the form
\begin{equation}\label{eq:cglax}
 L(\lda) = \frac{X_{0}^{(0)}}{\lda} + \frac{X_{0}^{(1)}}{\lda^2} + \frac{1}{T} \sum_{r=1}^N \sum_{k=0}^{T - 1} \frac{\sgm^{k} X_{r}}{\lda - \omg^k \zeta_r} + X_{\infty}
\end{equation}
with $\sgm X_{0}^{(j)}=\omega^{-j} X_{0}^{(j)}$ and $\sgm X_{\infty}=\omega X_{\infty}$.

Our goal is to give a Lagrangian multiform description of the cyclotomic Gaudin hierarchy. In a nutshell, what is required is to interpret $L(\lda)$ as an element of an appropriate coadjoint orbit so that we can write our Lagrangian multiform following the construction introduced in~\cite{CDS}. The latter is based on using the Lie dialgebra formalism \cite{STS} which is reviewed in Section~\ref{review_Lag_multi}. In Section~\ref{dialgebra_Gaudin}, we show how to construct the required algebraic ingredients for the cyclotomic Gaudin model. This is applied in Section~\ref{sec:cgmultiform}, where we obtain a Lagrangian multiform for the cyclotomic Gaudin hierarchy and discuss its properties.

Section~\ref{sec:realisations} deals with two special realisations of the cyclotomic Gaudin model: the periodic Toda chain (Section~\ref{sec:toda}) and the integrable case of the DST model (Section~\ref{sec:dst}). We derive Lagrangian multiforms for the hierarchies corresponding to these two theories using our results from Section~\ref{dialgebra_Gaudin}. The construction of the periodic Toda multiform, in particular, complements the results of \cite{CDS} where a Lagrangian multiform for the open Toda chain was derived, while in~\cite{SV} the infinite Toda chain was considered within a semi-discrete Lagrangian multiform. Finally, in Section~\ref{sec:todadst}, we illustrate how to couple the periodic Toda and DST hierarchies together in a~straightforward manner and derive a Lagrangian multiform for this coupled hierarchy. We end with concluding remarks in Section~\ref{sec:conclusion}.

\section{Lagrangian multiforms on coadjoint orbits}\label{review_Lag_multi}

In this section, we briefly review the results of \cite{CDS} where a Lagrangian multiform was obtained which produces a hierarchy of equations of motion in Lax form as its multi-time Euler--Lagrange equations. We will, in fact, be interested in applying this framework in the infinite-dimensional setting for which we refer more precisely to \cite[Section~7]{CDS}. In order to avoid confusion with the finite-dimensional Lie algebra $\g$ used throughout the rest of this paper, in this section we denote by $\mathbb{g}$ the infinite-dimensional Lie algebra of interest. The construction of the specific Lie algebra $\mathbb{g}$ relevant for the cyclotomic Gaudin model will be the subject of the next section.

The framework of Lie dialgebras for constructing hierarchies of compatible Lax equations involves a number of ingredients. For the purpose of this paper, we will actually only require a particular instance of this general framework corresponding to the so-called Adler--Kostant--Symes scheme \cite{A, K, Sy}, so we will mainly focus on this case to ease the comparison with later sections. The ingredients are:
\begin{enumerate}\itemsep=0pt
 \item %\label{pt:dialgone}
 A Lie algebra $\mathbb{g}$, with Lie group $\mathbb{G}$, which we assume is equipped with a nondegenerate invariant bilinear pairing
\begin{equation} \label{pairing V g}
 \langle \cdot, \cdot \rangle \colon \ V \times \mathbb{g} \longrightarrow \CC
\end{equation}
 with some representation $V$ of $\mathbb{G}$, allowing one to use $V$ as a model for the dual space $\mathbb{g}^\ast$ of $\mathbb{g}$. We will denote the representation of $\mathbb{G}$ by $\Ad^\ast \colon \mathbb{G} \times V \to V$, $(\varphi, \xi) \mapsto \Ad^\ast_\varphi \xi$ and the corresponding Lie algebra representation by $\ad^\ast \colon \mathbb{g} \times V \to V$, $(X, \xi) \mapsto \ad^\ast_X \xi$, and refer to these as the coadjoint representations.

 Note that if $\mathbb{g}$ were finite-dimensional then a good choice of representation space $V$ would be the algebraic dual $\mathbb{g}'$ of $\mathbb{g}$, in which case \eqref{pairing V g} is given by the canonical pairing. In the infinite-dimensional setting we are considering, the algebraic dual is too big. When $\mathbb{g}$ is equipped with a nondegenerate invariant symmetric bilinear form $(\cdot, \cdot) \colon \mathbb{g} \times \mathbb{g} \to \CC$ then one possible replacement for the algebraic dual is afforded by the smooth dual $(\mathbb{g}, \cdot) \subset \mathbb{g}'$ which is canonically isomorphic to $\mathbb{g}$ itself. However, in our present setting it turns out that the smooth dual will not be the appropriate notion of dual space, which is why we need the above more general working definition for the dual space $\mathbb{g}^\ast$.

 \item A linear map $R \colon \mathbb{g}\to\mathbb{g}$ satisfying the {\it modified} classical Yang--Baxter equation (mCYBE)
\begin{equation}
\label{mCYBE}
[R(X),R(Y)]-R\left([R(X),Y]+[X,R(Y)]\right)=-[X,Y] , \qquad \forall X,Y\in\mathbb{g} .
\end{equation}
This is the crucial ingredient providing us with the Lie dialgebra $(\mathbb{g},\mathbb{g}_R)$ where $\mathbb{g}_R$ is the vector space $\mathbb{g}$ equipped with a second Lie bracket $[X,Y]_R=\frac{1}{2}([R(X),Y]+[X,R(Y)])$. In the Adler--Kostant--Symes scheme, this is determined by a direct sum decomposition\footnote{Throughout the text, we use $\dotplus$ to denote a direct sum of vector spaces, and $\oplus$ to denote a direct sum of Lie algebras. This is a convenient convention borrowed from \cite{STS}.}
\begin{equation} \label{decomposition g}
\mathbb{g} = \mathbb{g}_+ \dotplus \mathbb{g}_-
\end{equation}
into complementary subalgebras $\mathbb{g}_\pm \subset \mathbb{g}$. A solution to \eqref{mCYBE} is then given by $R = P_+ - P_-$, where $P_\pm \colon \mathbb{g} \to \mathbb{g}_\pm$ denote the projections onto these respective subalgebras. In this case, the Lie algebra $\mathbb{g}_R$ is the direct sum of Lie algebras $\mathbb{g}_R = \mathbb{g}_+ \oplus \mathbb{g}_-$.
Now \eqref{decomposition g} induces a~corresponding direct sum decomposition of the model $V$ for the dual space $\mathbb{g}^\ast$, namely
\begin{equation} \label{decomposition V}
V = V_+ \dotplus V_-,
\end{equation}
where $V_\pm$ is the orthogonal complement of $\mathbb{g}_\pm$ with respect to the pairing \eqref{pairing V g}. We will denote by $\mathcal P_\pm \colon V \to V_\pm$ the projections onto these respective subspaces.

Let $\mathbb{G}_\pm$ denote the Lie subgroups of $\mathbb{G}$ associated with $\mathbb{g}_\pm$, and $\mathbb{G}_R$ denote the Lie group associated with $\mathbb{g}_R$, so that (at least locally) $\mathbb{G}_R \simeq \mathbb{G}_+ \times \mathbb{G}_-$. It will be useful to note for later that the coadjoint representation of $\mathbb{G}_R$ on $V \equiv \mathbb{g}^
\ast$ can be expressed in terms of that of the subgroups $\mathbb{G}_\pm \subset \mathbb{G}$ as
\begin{equation} \label{AdR}
\Ad^{R\ast}_\varphi \xi=\mathcal{P}_-\bigl(\Ad^\ast_{\varphi_+} \xi\bigr)+\mathcal{P}_+\bigl(\Ad^\ast_{\varphi_-}\xi\bigr)
\end{equation}
for any $\varphi = (\varphi_+, \varphi_-) \in \mathbb{G}_R$ and $\xi\in\mathbb{g}^\ast$.

\item A collection of independent $\Ad^\ast$-invariant functions on $\mathbb{g}^\ast$, say $H_k$, for $k=1,\dots,N$. It follows from their $\Ad^\ast$-invariance that these functions are in involution with respect to the Lie--Poisson bracket on $\mathbb{g}^\ast$ given by
 \begin{equation*}
%\label{Lie_PB}
 \{f,g\}_R(\xi)= \langle \xi,[\nabla f(\xi) , \nabla g(\xi)]_R \rangle , \qquad f,g\in C^\infty(\mathbb{g}^\ast),
\end{equation*}
and generate a hierarchy of compatible flows on the Poisson manifold $( \mathbb{g}^\ast, \{ \,,\, \}_R )$ which take the form of Lax equations
\begin{equation}
\label{hierarchy_Lax}
 \partial_{t_k}L := \{L,H_k\}_R = \ad_{M_k(L)}^\ast L ,\qquad M_k(L)=\frac{1}{2}R \nabla H_k(L) .
\end{equation}
When $\mathbb{g}$ and $V$ are subspaces of a common ambient Lie algebra and the coadjoint action $\ad^\ast \colon \mathbb{g} \times V \to V$ is given by a Lie bracket in this ambient Lie algebra, as will be the case for us in Section~\ref{dialgebra_Gaudin} below, the Lax equations \eqref{hierarchy_Lax} take the more familiar form
\begin{equation*}
%\label{hierarchy_Lax_com}
 \partial_{t_k}L = [M_k(L), L] .
\end{equation*}
This is the content of Semenov-Tian-Shansky's theorem, see, for instance, \cite{STS}.
\end{enumerate}

The symplectic leaves of the Poisson manifold $\bigl( \mathbb{g}^\ast, \{ \,,\, \}_R \bigr)$ are given by the coadjoint orbits
\begin{equation} \label{eq:coadj_orbit}
 {\cal O}_\Lambda=\bigl\{\Ad^{R\ast}_{\varphi} \Lambda \mid \varphi\in \mathbb{G}_R\bigr\} \qquad \text{for}\quad \Lambda \in \mathbb{g}^\ast ,
\end{equation}
which therefore provide natural phase spaces for integrable systems.
The key idea of \cite{CDS} was to define a Lagrangian multiform on such a coadjoint orbit which produces the hierarchy \eqref{hierarchy_Lax} as its Euler--Lagrange equations. Specifically, let $L = \Ad^{R\ast}_{\varphi} \Lambda$ and define
\begin{equation}
\label{our_Lag}
 \Lag[\varphi] = \sum_{k=1}^N \Lag_k \, \mathrm{d}t_k , \qquad \Lag_k= \bigl\langle L ,\,\partial_{t_k}\varphi \cdot_R \varphi^{-1} \bigr\rangle - H_k(L) ,
\end{equation}
where $\cdot_R$ denotes the multiplication in $\mathbb{G}_R$.
Then, we have the following.

\begin{Theorem}[\cite{CDS}]
\label{Th_multi_EL}
The Lagrangian $1$-form \eqref{our_Lag} satisfies the corner equations \eqref{simple_multitime_EL2}--\eqref{simple_multitime_EL3} of the multi-time Euler--Lagrange equations. The standard Euler--Lagrange equations \eqref{simple_multitime_EL1} associated with the Lagrangian coefficients $\Lag_k$ give the hierarchy of compatible Lax equations \eqref{hierarchy_Lax}%
 \begin{equation}
\label{EL_Lax}
 \partial_{t_k}L= \ad^\ast_{M_k(L)}L ,\qquad k=1,\dots,N .
\end{equation}
The closure relation\footnote{The closure relation for the Lagrangian one-form $\Lag$ is, in fact, equivalent to the involutivity of the Hamiltonians $H_k$ with respect to the Lie--Poisson $R$-bracket $\{\, ,\, \}_R$. We do not discuss this result here and refer the interested reader to \cite[Corollary 3.4]{CDS}.} holds: on solutions of \eqref{EL_Lax} we have
\[
 \partial_{t_k}\Lag_j-\partial_{t_j}\Lag_k=0 ,\qquad j,k=1,\dots,N .
\]
\end{Theorem}

In short, the machinery of \cite{CDS} provides a systematic way of constructing a Lagrangian multiform for a hierarchy of Lax equations \eqref{hierarchy_Lax} on a coadjoint orbit \eqref{eq:coadj_orbit}. To apply this result to the cyclotomic Gaudin model we must identify its Lax matrix \eqref{eq:cglax} as an element of an~appropriate coadjoint orbit ${\cal O}_{\Lambda} \subset \mathbb{g}^\ast$, as in \eqref{eq:coadj_orbit}, for a suitable Lie algebra $\mathbb{g}$ and $R$-matrix. In~particular, we will need to build a suitable model $V$ for the dual space $\mathbb{g}^\ast$ via a corresponding pairing $\langle \cdot, \cdot \rangle \colon V \times \mathbb{g} \to \mathbb{C}$. This will be the central goal of the next section.

\section{Cyclotomic Gaudin model}\label{dialgebra_Gaudin}

The purpose of this section is to formulate a Lagrangian multiform for the cyclotomic Gaudin model using the machinery developed in \cite{CDS}.
As outlined at the end of Section~\ref{review_Lag_multi}, the first immediate task is to introduce the relevant infinite-dimensional Lie algebra $\mathbb{g}$ and a suitable model $V$ for its dual space $\mathbb{g}^\ast$ by constructing a nondegenerate invariant bilinear pairing as in~\eqref{pairing V g}. This will be achieved in Section~\ref{sec:Lie dialgebra Gaudin}, which culminates in a description of the relevant coadjoint orbits $\mathcal O_\Lambda \subset \mathbb{g}^\ast$ in Lemma \ref{identification_orbits}. In Section~\ref{sec:cgmlaxm}, we specialise to the case of simple poles at $\omg^k \zeta_r$, $k \in \{0, \dots, T-1\}$, $r \in \{1, \dots, N\}$, and double poles at the origin and at infinity, and then describe the Lax matrix \eqref{eq:cglax} of the cyclotomic Gaudin model as an element of this coadjoint orbit. In Section~\ref{sec:cgmlaxeq}, we then describe the Lax equations as flows on this coadjoint orbit, as in \eqref{hierarchy_Lax}. Finally, in Section~\ref{sec:cgmultiform}, we put everything together to obtain the Lagrangian multiform for the cyclotomic Gaudin model.

\subsection{Lie dialgebra framework}\label{sec:Lie dialgebra Gaudin}

Let us fix $T \in \mathbb{Z}_{\geq 1}$ and pick a primitive $T$-th root of unity $\omg \in \mathbb{C}^{\times}$. We can then define
\begin{equation*}
 \Gamma \coloneqq \bigl\{ 1, \omg, \omg^2, \dots, \omg^{T-1} \bigr\} ,
\end{equation*}
a copy of the cyclic group $\mathbb{Z}_T$ of order $T$ that acts on $\mathbb{C}$ by multiplication. Further, we introduce the finite set of points $D = \{0, \zeta_1, \dots, \zeta_N, \infty\} \subset \mathbb{C}P^1$ including the point at infinity such that the $\Gamma$-orbits of the points $\zeta_1, \dots, \zeta_N$ are pairwise disjoint, that is,
\begin{equation*}
 \Gamma \zeta_r \cap \Gamma \zeta_s = \varnothing \qquad \text{for all} \quad 1 \leq r \neq s \leq N .
\end{equation*}
Note that unlike the non-zero finite points $\zeta_1, \dots, \zeta_N$, the origin and the point at infinity are fixed under the action of $\Gamma$.

Let $\g$ be a finite-dimensional Lie algebra over $\mathbb{C}$ with an automorphism $\sgm$ of order $T$. For simplicity, we will only consider matrix Lie algebras in this work, with the Lie bracket being the commutator, and a nondegenerate invariant bilinear pairing given by the trace. We also assume that the automorphism $\sgm$ preserves the trace, that is, $\Tr(\sgm (x) \sgm (y)) = \Tr(xy)$ for any $x, y \in \g$. The eigenspaces of $\sgm$,
\begin{equation}\label{eq:sgmeigenspace}
 \g^{(k)} = \bigl\{X \in \g \mid \sgm(X) = \omg^k X \bigr\}, \qquad k \in \{0, \dots, T-1\} ,
\end{equation}
form a $\mathbb{Z}_T$-gradation of $\g$, namely
\begin{equation}
\label{gradation}
 \g = \g^{(0)} \dotplus \dots \dotplus \g^{(T-1)} \qquad \text{with} \quad \bigl[\g^{(k_1)}, \g^{(k_2)}\bigr] = \g^{(k_1 + k_2\, \text{mod}\, T)} .
\end{equation}
From now on, we will simply write $\g^{(n)}$ to mean $\g^{(n\, \text{mod}\, T)}$. Further, let us define the local parameters
\begin{equation} \label{lambda r def}
 \lda_0 = \lda , \qquad \lda_r = \lda - \zeta_r \quad \text{for}\ r \in \{1, \dots, N\} , \qquad \lda_\infty = \dfrac{1}{\lda} .
\end{equation}
It will be convenient to also introduce an additional set of local parameters $\lda_{r, k} = \lda - \omg^k \zeta_r$ for all $k \in \{0, \dots, T-1\}$ and $r \in \{1, \dots, N\}$. We also introduce an index set $S = \{0, 1, \dots, N, \infty\}$ where $0$ and $\infty$ denote labels for indices rather than points in $\mathbb{C}P^1$.

Let us denote by $\mathcal{F}_{D^{\prime}}$ the algebra of rational functions in the formal variable $\lda$ with values in $\g$ that are regular outside $D^{\prime} = \bigl\{ \omg^k \zeta_r \mid k \in \{0, \dots, T-1\},\, r \in \{1, \dots, N\}\bigr\} \cup \{0, \infty\}$. Two subspaces of this algebra will be of relevance here: the subspace $\mathcal{F}_{D^\prime}^\Gamma$ of \emph{equivariant functions} and the subspace $\Omega_{D^\prime}^\Gamma$ of \emph{equivariant one-forms} that we define as
\begin{gather*}
 \mathcal{F}_{D^\prime}^{\Gamma} \coloneqq \{ f \in \mathcal{F}_{D^\prime} \mid \sgm (f (\lda)) = f (\omg \lda) \} ,\qquad
 \Omega_{D^\prime}^{\Gamma} \coloneqq \{ g \in \mathcal{F}_{D^\prime} \mid \sgm (g (\lda)) = \omg g (\omg \lda) \} .
 \end{gather*}
In general, equivariant functions and equivariant one-forms take the forms
\begin{gather*}
 f = \sum_{n=0}^{M_0 - 1} \frac{X_0^{(n)}}{\lda^{n+1}} + \sum_{r=1}^{N} \sum_{k=0}^{T-1} \sum_{n=0}^{M_r - 1} \frac{\sgm^k X_r^{(n)}}{\bigl(\omg^{-k}\lda - \zeta_r\bigr)^{n+1}} + \sum_{n=0}^{M_\infty} \frac{X_\infty^{(n)}}{\lda_\infty^n} ,\\
 g = \sum_{n=0}^{N_0 - 1} \frac{Y_0^{(n)}}{\lda^{n+1}} + \sum_{r=1}^N \sum_{k=0}^{T-1} \sum_{n=0}^{N_r - 1} \frac{\omg^{-k} \sgm^k Y_r^{(n)}}{\bigl(\omg^{-k} \lda - \zeta_r\bigr)^{n+1}} + \sum_{n=0}^{N_\infty} \frac{Y_\infty^{(n)}} {\lda_\infty^n} ,
 \end{gather*}
respectively, with
\begin{alignat}{4}
 &X_0^{(n)} \in \g^{(-n-1)},\qquad&& X_r^{(n)} \in \g \quad \text{for}\ r \in \{1, \dots, N\},\qquad&& X_\infty^{(n)} \in \g^{(n)} , &\nonumber\\
 &Y_0^{(n)} \in \g^{(-n)},\qquad&& Y_r^{(n)} \in \g \quad \text{for}\ r \in \{1, \dots, N\}, \qquad&& Y_\infty^{(n)} \in \g^{(n+1)} .&\label{eq:coeffcond}
\end{alignat}

One can check that the subspace $\mathcal{F}_{D^\prime}^\Gamma$ is, in fact, a Lie subalgebra of $\mathcal{F}_{D^\prime}$. The subspace $\Omega_{D^\prime}^\Gamma$, on the other hand, contains the Lax matrix \eqref{eq:cglax} of the model. To construct the cyclotomic Gaudin multiform \`{a} la \cite{CDS}, we will realise this Lax matrix as an element of a certain coadjoint orbit. To do so, let us define the loop algebra of formal Laurent series in the local parameter $\lda_r$ with coefficients in $\g$, for each $r \in S$,
\begin{equation} \label{Lgr def}
 \Lg_r = \g \otimes \mathbb{C}((\lda_r)) ,
\end{equation}
with Lie bracket
\begin{equation} \label{Lie bracket Lg}
 \bigl[X\lda_r^i, Y\lda_r^j\bigr] = [X, Y]\lda_r^{i+j}, \qquad X, Y \in \g ,
\end{equation}
and then consider the direct sum
\begin{equation*}
 \Lg_D = \bigoplus_{r \in S} \Lg_r .
\end{equation*}

The Lie bracket of two elements $X, Y \in \Lg_D$ is defined component-wise
\begin{equation*}
 [X,Y]_r(\lda_r)= [X_r(\lda_r), Y_r(\lda_r)] .
\end{equation*}
Of interest to us here are two particular subspaces of $\Lg_D$,
\begin{equation*}
%\label{def_g_V}
 \Lg_{D}^{(0)} = \Lg_0^{\Gamma, 0} \oplus \bigoplus_{r=1}^N \Lg_r \oplus \Lg_\infty^{\Gamma, 0} \qquad \text{and} \qquad \Lg_{D}^{(1)} = \Lg_0^{\Gamma, 1} \oplus \bigoplus_{r=1}^N \Lg_r \oplus \Lg_\infty^{\Gamma, 1} ,
\end{equation*}
where the twisted spaces attached to the origin and infinity are defined as
\begin{align}
 &\Lg_0^{\Gamma, k} = \bigl\{X_0(\lda_0) \in \Lg_0 \mid \sgm (X_0(\lda_0)) = \omg^k X_0(\omg \lda_0)\bigr\} ,\nonumber\\
 &\Lg_\infty^{\Gamma, k} = \bigl\{X_\infty(\lda_\infty) \in \Lg_\infty \mid \sgm (X_\infty(\lda_\infty)) = \omg^k X_\infty\bigl(\omg^{-1} \lda_\infty\bigr)\bigr\} \label{Lg Gamma spaces}
 \end{align}
for $k = 0, 1$. Elements of \smash{$\Lg_{D}^{(0)}$} and \smash{$\Lg_{D}^{(1)}$} are tuples of the form
\begin{equation} \label{X g0D explicit}
 X = \left(\sum_{n=-\infty}^{M_0 - 1} \frac{X_0^{(n)}}{\lda^{n+1}}, \sum_{n=-\infty}^{M_1 - 1} \frac{X_1^{(n)}}{(\lda - \zeta_1)^{n+1}}, \dots, \sum_{n=-\infty}^{M_N - 1} \frac{X_N^{(n)}}{(\lda - \zeta_N)^{n+1}}, \sum_{n=-\infty}^{M_\infty} \frac{X_\infty^{(n)}}{\lda_\infty^n} \right) ,
\end{equation}
and
\begin{equation} \label{Y g1D explicit}
 Y = \left(\sum_{n=-\infty}^{N_0 - 1} \frac{Y_0^{(n)}}{\lda^{n+1}}, \sum_{n=-\infty}^{N_1 - 1} \frac{Y_1^{(n)}}{(\lda - \zeta_1)^{n+1}}, \dots, \sum_{n=-\infty}^{N_N - 1} \frac{Y_N^{(n)}}{(\lda - \zeta_N)^{n+1}}, \sum_{n=-\infty}^{N_\infty} \frac{Y_\infty^{(n)}}{\lda_\infty^n} \right) ,
 \end{equation}
 respectively, where the coefficients $X_r^{(n)}$ and $Y_r^{(n)}$, $r \in S$, satisfy the conditions in \eqref{eq:coeffcond}. Note that each coefficient in \eqref{X g0D explicit} and \eqref{Y g1D explicit} is a $\g$-valued Laurent series, i.e., an element of \eqref{Lgr def} for~${r \in S}$. This is clear from changing variables in each sum from $n$ to $-n$. The reason for the unusual choice of range for the indices $n$ in each of the above sums is to make the expressions for the maps in \eqref{pi01 def} defined below slightly more transparent.

 The subspace $\Lg_{D}^{(0)}$ is a Lie subalgebra of $\Lg_D$ which defines the desired infinite-dimensional~Lie algebra \smash{$\mathbb{g} = \Lg_D^{(0)}$} in the notation of Section~\ref{review_Lag_multi}. Continuing the identification of the different ingredients from the Lie dialgebra framework listed in Section~\ref{review_Lag_multi}, notice that \smash{$V = \Lg_D^{(1)}$} is clearly a~representation of the Lie algebra \smash{$\mathbb{g} = \Lg_D^{(0)}$} since we have \smash{$\bigl[\Lg_r^{\Gamma, 0}, \Lg_r^{\Gamma, 1}\bigr] \subset \Lg_r^{\Gamma, 1}$} for $r \in \{ 0, \infty \}$. In the notation of Section~\ref{review_Lag_multi}, the representation $\ad^\ast \colon \mathbb{g} \times V \to V$ is given explicitly by $(X, Y) \mapsto [X, Y]$. The next proposition identifies \smash{$V = \Lg_D^{(1)}$} as a suitable model for the dual space $\mathbb{g}^\ast$.

\begin{Proposition}\label{prop:pairing}
 The bilinear pairing \smash{$\langle \cdot, \cdot \rangle \colon \Lg_D^{(1)} \times \Lg_D^{(0)} \to \CC$} defined by
 \begin{gather}\label{eq:bmgd}
 \langle Y, X \rangle = T \sum_{r=1}^{N} \Res_{\lda_r =0} \Tr(Y_r(\lda_r) X_r(\lda_r)) {\rm d}\lda + \sum_{r \in \{0, \infty\}} \Res_{\lda_r =0} \Tr(Y_r(\lda_r) X_r(\lda_r)) {\rm d}\lda
\end{gather}
for any \smash{$Y \in \Lg_{D}^{(1)}$} and \smash{$X \in \Lg_{D}^{(0)}$}, is nondegenerate and invariant under the adjoint action of \smash{$\Lg_D^{(0)}$}.
\end{Proposition}

\begin{proof}
Recall that the trace $\Tr \colon \g \times \g \to \CC$, $(x,y) \mapsto \Tr(xy)$ is a nondegenerate invariant bilinear pairing on $\g$ and invariant under the action of $\sgm$. Given any $x \in \g^{(m)}$ and $y \in \g^{(n)}$, we have $\Tr(xy) = \Tr(\sgm (x) \sgm (y)) = \omg^{m+n} \Tr(xy)$. Therefore, $\Tr(xy) = 0$ for all $x \in \g^{(m)}$ and $y \in \g^{(n)}$ if $m+n \neq 0 \mod T$. Now, since $\Tr \colon \g \times \g \to \CC$ is nondegenerate on $\g$, given $y \in \g^{(n)}$, there is $x\in \g$ such that $\Tr(xy)\neq 0$. From \eqref{gradation}, we have $x=x^{(0)}+\dots+x^{(T-1)}$, and from the previous result $\Tr(xy)=\Tr\bigl(x^{(m)}y\bigr)$ with $m+n= 0\mod T$. Hence, $\Tr\bigl(x^{(m)}y\bigr)\neq 0$ and we get a nondegenerate pairing between the subspaces \smash{$\g^{(m)}$} and \smash{$\g^{(n)}$} with $m+n = 0\mod T$.

Now given any non-zero element $Y \in \Lg_{D}^{(1)}$, it has a non-zero component $Y_r^{(m)} \lambda_r^{-m-1}$ for some $m \in \ZZ$ and $r \in S$. If $r=0$ then \smash{$Y_0^{(m)} \in \g^{(-m)}$} by \eqref{eq:coeffcond}, and we can find an \smash{$X_0^{(-m-1)} \in \g^{(m)}$} which pairs non-trivially with it under the trace. So, \smash{$X_0^{(-m-1)} \lambda^m$} pairs non-trivially with $Y$. Likewise, if $r=\infty$ then \smash{$Y_\infty^{(m)} \in \g^{(m+1)}$} by \eqref{eq:coeffcond} and we can find an \smash{$X_\infty^{(-m-1)} \in \g^{(-m-1)}$} which pairs non-trivially with it under the trace so that \smash{$X_\infty^{(-m-1)} \lambda_\infty^{m+1}$} pairs non-trivially with $Y$. And if $r \in \{1,\dots, N\}$ then pick any \smash{$X_r^{(-m-1)} \in \g$} which pairs non-trivially with \smash{$Y_r^{(m)} \in \g$} under the trace so that \smash{$X_r^{(-m-1)} \lambda_r^m$} pairs non-trivially with $Y$. Therefore, the bilinear pairing $\langle \cdot, \cdot \rangle$ is nondegenerate in the first argument. A similar argument establishes the nondegeneracy in the second argument.

Finally, the invariance of the bilinear pairing $\langle \cdot, \cdot \rangle$ under the adjoint action of \smash{$\Lg_D^{(0)}$} follows from the definition \eqref{Lie bracket Lg} of the Lie bracket in $\Lg_D$ and the invariance of the trace under the adjoint action of $\g$.
\end{proof}

Next, we turn to the identification of the subalgebras $\mathbb{g}_\pm \subset \mathbb{g}$ and the corresponding subspaces $V_\pm \subset V$ in the notation of Section~\ref{review_Lag_multi}.
Crucially for us, the spaces $\mathcal{F}_{D^\prime}^\Gamma$ and $\Omega_{D^\prime}^\Gamma$ embed into \smash{$\Lg_{D}^{(0)}$} and \smash{$\Lg_{D}^{(1)}$}, respectively,
\begin{align}
 &\iota_\lda \colon \ \mathcal{F}_{D^\prime}^\Gamma \longhookrightarrow \Lg_D^{(0)} , \qquad f \longmapsto (\iota_{\lda_0}f, \iota_{\lda_1}f, \dots, \iota_{\lda_N}f, \iota_{\lda_\infty}f) ,\nonumber\\
 &\iota_\lda \colon \ \Omega_{D^\prime}^\Gamma \longhookrightarrow \Lg_D^{(1)} , \qquad g \longmapsto (\iota_{\lda_0}g, \iota_{\lda_1}g, \dots, \iota_{\lda_N}g, \iota_{\lda_\infty}g) ,\label{iota map def}
\end{align}
where, for each $r \in S$, $\iota_{\lda_r}f$ and $\iota_{\lda_r}g$ respectively denote the formal Laurent expansion of~$f$ and~$g$ about the point $\zeta_r$.
It will also be useful to define left inverses for these embeddings:
\begin{align}
 &\pi_\lda^{(0)} \colon \ \Lg_D^{(0)} \longrightarrow \mathcal{F}_{D^\prime}^\Gamma ,\nonumber\\
 &\pi_\lda^{(1)} \colon \ \Lg_D^{(1)} \longrightarrow \Omega_{D^\prime}^\Gamma ,\label{pi01 def}
\end{align}
defined explicitly as
\begin{align*}%\label{eq:maptoequivfunc}
 &\pi_\lda^{(0)} \Biggl(\sum_{n=-\infty}^{M_0 - 1} \frac{X_0^{(n)}}{\lda^{n+1}}, \sum_{n=-\infty}^{M_1 - 1} \frac{X_1^{(n)}}{(\lda - \zeta_1)^{n+1}}, \dots, \sum_{n=-\infty}^{M_N - 1} \frac{X_N^{(n)}}{(\lda - \zeta_N)^{n+1}}, \sum_{n=-\infty}^{M_\infty} \frac{X_\infty^{(n)}}{\lda_\infty^n} \Biggr) \notag \\
 & \qquad{} = \sum_{n=0}^{M_0 - 1} \frac{X_0^{(n)}}{\lda^{n+1}} + \sum_{r=1}^N \sum_{k=0}^{T-1} \sum_{n = 0}^{M_r - 1} \frac{\sgm^k X_r^{(n)}}{(\omg^{-k}\lda - \zeta_r)^{n+1}} + \sum_{n=0}^{M_\infty} \frac{X_\infty^{(n)}}{\lda_\infty^n} ,
\end{align*}
and
\begin{align*}
 &\pi_\lda^{(1)} \Biggl(\sum_{n=-\infty}^{N_0 - 1} \frac{Y_0^{(n)}}{\lda^{n+1}}, \sum_{n=-\infty}^{N_1 - 1} \frac{Y_1^{(n)}}{(\lda - \zeta_1)^{n+1}}, \dots, \sum_{n=-\infty}^{N_N - 1} \frac{Y_N^{(n)}}{(\lda - \zeta_N)^{n+1}}, \sum_{n=-\infty}^{N_\infty} \frac{Y_\infty^{(n)}}{\lda_\infty^n} \Biggr) \notag \\
 & \qquad{}= \sum_{n=0}^{N_0 - 1} \frac{Y_0^{(n)}}{\lda^{n+1}} + \sum_{r=1}^N \sum_{k=0}^{T-1} \sum_{n=0}^{N_r - 1} \frac{\omg^{-k} \sgm^k Y_r^{(n)}}{(\omg^{-k} \lda - \zeta_r)^{n+1}} + \sum_{n=0}^{N_\infty} \frac{Y_\infty^{(n)}}{\lda_\infty^n} .
\end{align*}
We will also need to define the following subspaces of $\Lg_D^{(0)}$ and $\Lg_D^{(1)}$, respectively:
\begin{equation*}
 \Lg_{D+}^{(0)} = \Lg_{0+}^{\Gamma, 0} \oplus \bigoplus_{r=1}^N \Lg_{r+} \oplus \Lg_{\infty+}^{\Gamma, 1} \qquad \text{and} \qquad \Lg_{D+}^{(1)} = \Lg_{0+}^{\Gamma, 1} \oplus \bigoplus_{r=1}^N \Lg_{r+} \oplus \Lg_{\infty+}^{\Gamma, 1},
\end{equation*}
where for $k = 0, 1$, we introduced (cf. \eqref{Lgr def} and \eqref{Lg Gamma spaces})
\begin{align*}
 &\Lg_{0+}^{\Gamma, k} = \Lg_{0}^{\Gamma, k} \cap \g \otimes \mathbb{C}[[\lda_0]] ,\\
 &\Lg_{r+} = \g \otimes \mathbb{C}[[\lda_r]] , \qquad r \in \{1, \dots, N \} ,\\
 &\Lg_{\infty+}^{\Gamma, k} = \Lg_{\infty}^{\Gamma, k} \cap \g \otimes \lda_\infty \mathbb{C}[[\lda_\infty]] .
\end{align*}
Here we denoted by $\g \otimes \mathbb{C}[[\lda_r]]$ the algebra of formal Taylor series in $\lda_r$ with coefficients in $\g$, for $r \neq \infty$, and by $\g \otimes \lda_\infty \mathbb{C}[[\lda_\infty]]$ the algebra of formal Taylor series in $\lda_\infty$ with coefficients in~$\g$ without constant term.

Coming back to the identification of the Lie dialgebra ingredients from Section~\ref{review_Lag_multi}, the next proposition identifies the desired decomposition \eqref{decomposition g} of \smash{$\mathbb{g} = \Lg_D^{(0)}$} into complementary subalge\-bras~${\mathbb{g}_\pm \subset \mathbb{g}}$. Explicitly, we have the identifications \smash{$\mathbb{g}_+ = \Lg_{D+}^{(0)}$} and \smash{$\mathbb{g}_- = \iota_\lda \mathcal{F}_{D^\prime}^\Gamma$}. We also identify a~decomposition \eqref{decomposition V} of our model \smash{$V = \Lg_D^{(1)}$} for the dual space $\mathbb{g}^\ast$ into complementary subspaces $V_\pm \subset V$, where explicitly \smash{$V_+ = \Lg_{D+}^{(1)}$} and $V_- = \iota_\lda \Omega_{D^\prime}^\Gamma$, but will show only later in Proposition~\ref{prop:pairequivsp} that this decomposition of $V$ is the desired one induced by that of $\mathbb{g}$.

\begin{Proposition}\label{prop:gddecomp}
 The spaces \smash{$\Lg_D^{(0)}$} and \smash{$\Lg_D^{(1)}$} admit the vector space decompositions
 \begin{align}
 \label{eq:gdzerodecomp}
 &\Lg_D^{(0)} = \Lg_{D+}^{(0)} \dotplus \iota_\lda \mathcal{F}_{D^\prime}^\Gamma ,\\
 \label{eq:gdonedecomp}
 &\Lg_D^{(1)} = \Lg_{D+}^{(1)} \dotplus \iota_\lda \Omega_{D^\prime}^\Gamma ,
 \end{align}
 respectively. Moreover, the subspaces \smash{$\Lg_{D+}^{(0)}$} and \smash{$\iota_\lda \mathcal{F}_{D^\prime}^\Gamma$} are Lie subalgebras of \smash{$\Lg_D^{(0)}$}.
\end{Proposition}
\begin{proof}
 To any $X \!=\! (X_0, X_1, \dots, X_N, X_\infty) \!\in\! \Lg_D^{(0)}$, we associate an equivariant function ${f \!=\! \pi_\lda^{(0)}X}$. Notice that for all $r \in S$, $X_r - \iota_{\lda_r} f$ is a formal Taylor series in $\lda_r$. It follows that $X$ splits uniquely as the direct sum of the tuples \smash{$(X_0 - \iota_{\lda_0} f, X_1 - \iota_{\lda_1} f, \dots, X_N - \iota_{\lda_N}f, X_\infty - \iota_{\lda_\infty} f) \in \Lg_{D+}^{(0)}$} and $(\iota_{\lda_0} f, \iota_{\lda_1} f, \dots, \iota_{\lda_N}f, \iota_{\lda_\infty} f) \in \iota_\lda \mathcal{F}_{D^\prime}^\Gamma$, as required.

 One can repeat the above steps \big(with the map $\pi_\lda^{(1)}$ acting on some $Y \in \Lg_D^{(1)}$\big) to prove that~\eqref{eq:gdonedecomp} defines a vector space decomposition as well.
\end{proof}

Let $P_+$ and $P_-$ denote the projectors onto the subspaces $\Lg_{D+}^{(0)}$ and $\iota_\lda \mathcal{F}_{D^\prime}^\Gamma$ respectively, relative~to the decomposition in \eqref{eq:gdzerodecomp}, and $\P_+$ and $\P_-$ the projectors onto \smash{$\Lg_{D+}^{(1)}$} and \smash{$\iota_\lda \Omega_{D^\prime}^\Gamma$} respectively, relative to the decomposition in \eqref{eq:gdonedecomp}, in line with the notation from Section~\ref{review_Lag_multi}. As~recalled in~Section~\ref{review_Lag_multi}, the linear map
\begin{equation} \label{R def}
 R = P_+ - P_-
\end{equation}
is a solution of the mCYBE on \smash{$\Lg_{D}^{(0)}$}. It is also useful to define the maps
\begin{equation*}
 R_\pm = \frac{1}{2}(R \pm \text{id}),
\end{equation*}
which can be related to the projectors onto \smash{$\Lg_{D+}^{(0)}$} and \smash{$\iota_\lda \mathcal{F}_{D^\prime}^\Gamma$} relative to direct sum decomposition~\eqref{eq:gdzerodecomp} as
\begin{equation*}
 R_\pm = \pm P_\pm .
\end{equation*}
The linear map $R$ is related to the non-skew-symmetric $r$-matrix \eqref{eq:nss-rmat} underlying the cyclotomic Gaudin model, as one would naturally expect from the construction. More precisely, the expression \eqref{eq:nss-rmat} provides the kernel of the linear map \eqref{R def} with respect to the bilinear pairing \eqref{eq:bmgd}. To show this, we will use standard tensor product space notation as follows. Given any rational function $U_{\ti{12}}(\lambda, \mu)$ in the parameters $\lda$ and $\mu$ such that \smash{$\iota_\mu \iota_\lda U_{\ti{12}}, \iota_\lda \iota_\mu U_{\ti{12}} \in \Lg_D^{(0)}\, \widetilde{\otimes}\, \Lg_D^{(1)}$}, where the first tensor factor in the formally completed tensor product is the loop algebra \smash{$\Lg_D^{(0)}$} with the loop parameter $\lambda$, and the second tensor factor is the space \smash{$\Lg_D^{(1)}$} with the loop parameter $\mu$, we introduce the following shorthand notations:
\begin{align*}%\label{eq:bmgdtensorone}
 {\bigl\langle \iota_\mu \iota_\lda U_{\ti{12}}, X_{\ti{2}} \bigr\rangle}_{\ti{2}} &{}= T \sum_{r=1}^N \Res_{\mu_r = 0} \Tr_{\ti{2}} \bigl( \iota_{\mu_r} \iota_\lda U_{\ti{12}}(\lda, \mu) X_{r \ti{2}}(\mu_r)\bigr){\rm d}\mu \notag \\
 &\quad{} + \sum_{r \in \{0, \infty\}} \Res_{\mu_r = 0} \Tr_{\ti{2}} \bigl( \iota_{\mu_r} \iota_\lda U_{\ti{12}}(\lda, \mu) X_{r \ti{2}}(\mu_r)\bigr){\rm d}\mu ,
\end{align*}
and
\begin{align*}%\label{eq:bmgdtensortwo}
 {\bigl\langle \iota_\lda \iota_\mu U_{\ti{12}}, X_{\ti{2}} \bigr\rangle}_{\ti{2}} &{}= T \sum_{r=1}^N \Res_{\mu_r = 0} \Tr_{\ti{2}} \bigl( \iota_\lda \iota_{\mu_r} U_{\ti{12}}(\lda, \mu) X_{r \ti{2}}(\mu_r)\bigr){\rm d}\mu \notag \\
 &\quad{} + \sum_{r \in \{0, \infty\}} \Res_{\mu_r = 0} \Tr_{\ti{2}} \bigl( \iota_\lda \iota_{\mu_r} U_{\ti{12}}(\lda, \mu) X_{r \ti{2}}(\mu_r)\bigr){\rm d}\mu
\end{align*}
for any \smash{$X \in \Lg_D^{(0)}$}, where the parameters $\mu_r$ are defined analogously to the parameters $\lda_r$ in \eqref{lambda r def}, and the linear map $\iota_\mu$ is defined as in \eqref{iota map def}, returning the tuple of formal Laurent expansions in $\mu_r$, $r \in S$.
\begin{Proposition}\label{prop:projrel}
For all \smash{$X \in \Lg_D^{(0)}$}, the linear maps $R_+$ and $R_-$ satisfy
\begin{equation}\label{eq:projrel}
 R_+(X) = {\bigl\langle \iota_\mu \iota_\lda r_{\ti{12}}, X_{\ti{2}} \bigr\rangle}_{\ti{2}} \qquad \text{and} \qquad R_-(X) = {\bigl\langle \iota_\lda \iota_\mu r_{\ti{12}}, X_{\ti{2}} \bigr\rangle}_{\ti{2}},
\end{equation}
respectively, where
\begin{equation*}
 r_{\ti{12}}(\lda, \mu) = \frac{1}{T} \sum_{k=0}^{T-1} \frac{\sgm_{\ti{1}}^k C_{\ti{12}}}{\mu - \omg^{-k}\lda} .
\end{equation*}
\end{Proposition}
\begin{proof}
In what follows, to treat the origin on the same footing as the points $\zeta_r$, $r \in \{1, \dots, N \}$, it will be convenient to introduce the notation $\zeta_0 = 0$. We have
 \begin{gather}
 \iota_{\lda_s} \Biggl( \frac{1}{T} \sum_{k=0}^{T-1} \frac{\sgm_{\ti{1}}^k C_{\ti{12}}}{\mu - \omg^{-k}\lda} \Biggr) \nonumber\\
 \qquad{} = \begin{dcases}
 \frac{1}{T} \sum_{k=0}^{T-1} \sum_{m=0}^{\infty} \frac{\omg^{-km}(\lda - \zeta_s)^m \sgm_{\ti{2}}^{-k} C_{\ti{12}}}{\bigl(\mu - \omg^{-k}\zeta_s\bigr)^{m+1}} &\text{for}\ s \in \{0, 1, \dots, N \},\\
 - \frac{1}{T} \sum_{k=0}^{T-1} \sum_{m=0}^{\infty} \frac{\omg^{k(m+1)} \mu^m \sgm_{\ti{2}}^{-k} C_{\ti{12}}}{\lda^{m+1}} &\text{for}\ s = \infty ,\\
 \end{dcases}\label{eq:expldar}
 \end{gather}
 and
 \begin{gather}
 \iota_{\mu_r} \Biggl( \frac{1}{T} \sum_{k=0}^{T-1} \frac{\sgm_{\ti{1}}^k C_{\ti{12}}}{\mu - \omg^{-k}\lda} \Biggr)\nonumber\\
 \qquad{} = \begin{dcases}
 - \frac{1}{T} \sum_{k=0}^{T-1} \sum_{m=0}^{\infty} \frac{\omg^{k(m+1)}(\mu - \zeta_r)^m \sgm_{\ti{1}}^k C_{\ti{12}}}{\bigl(\lda - \omg^k\zeta_r\bigr)^{m+1}} &\text{for}\ r \in \{0, 1, \dots, N \},\\
 \frac{1}{T} \sum_{k=0}^{T-1} \sum_{m=0}^{\infty} \frac{\omg^{-km} \lda^m \sgm_{\ti{1}}^k C_{\ti{12}}}{\mu^{m+1}} &\text{for}\ r = \infty .\\
 \end{dcases}\label{eq:expmur}
 \end{gather}
It follows that both $\iota_\mu \iota_\lda r_{\ti{12}}$ and \smash{$\iota_\lda \iota_\mu r_{\ti{12}}$} are elements of \smash{$\Lg_D^{(0)} \, \widetilde{\otimes}\, \Lg_D^{(1)}$} where the loop parameter~$\lambda$~is in the first tensor factor and $\mu$ in the second. Now, pick an arbitrary \smash{$X \in \Lg_D^{(0)}$} and let $X = W+Z$ be its decomposition relative to \eqref{eq:gdzerodecomp} where the two components can be written explicitly as
 \begin{equation*}
 W = \Biggl( \sum_{n=0}^\infty W_0^{(n)} \lda^n, \sum_{n=0}^\infty W_1^{(n)} (\lda - \zeta_1)^n, \dots, \sum_{n=0}^\infty W_N^{(n)} (\lda - \zeta_N)^n, \sum_{n=1}^\infty W_\infty^{(n)}\lda_\infty^n \Biggr) \in \Lg_{D+}^{(0)} ,
 \end{equation*}
 and
 \begin{equation*}
 Z = \iota_\lda \Biggl( \sum_{n=0}^{N_0 - 1} \frac{Z_0^{(n)}}{\lda^{n+1}} + \sum_{r=1}^{N} \sum_{k=0}^{T-1} \sum_{n=0}^{N_r - 1} \frac{\sgm^k Z_r^{(n)}}{\bigl(\omg^{-k}\lda - \zeta_r\bigr)^{n+1}} + \sum_{n=0}^{N_\infty} \frac{Z_\infty^{(n)}}{\lda_\infty^n} \Biggr) \in \iota_\lda \mathcal{F}_{D^\prime}^\Gamma .
 \end{equation*}
Using the expansion \eqref{eq:expldar}, we get
\begin{align*}
 &T \sum_{r=1}^N \Res_{\mu_r = 0} \Tr_{\ti{2}} \bigl( \iota_{\mu_r} \iota_{\lda_s} r_{\ti{12}}(\lda, \mu) W_{r \ti{2}}(\mu_r)\bigr){\rm d}\mu \notag\\
 &\qquad\quad{}+ \sum_{r \in \{0, \infty\}} \Res_{\mu_r = 0} \Tr_{\ti{2}} \bigl( \iota_{\mu_r} \iota_{\lda_s} r_{\ti{12}}(\lda, \mu) W_{r \ti{2}}(\mu_r)\bigr){\rm d}\mu \notag \\
 &\qquad{}= \sum_{m=0}^\infty W_s^{(m)} (\lda - \zeta_s)^m, \qquad \text{when} \quad s \in \{0, 1, \dots, N\},
\end{align*}
and
\begin{align*}
 &T \sum_{r=1}^N \Res_{\mu_r = 0} \Tr_{\ti{2}} \bigl( \iota_{\mu_r} \iota_{\lda_s} r_{\ti{12}}(\lda, \mu) W_{r \ti{2}}(\mu_r)\bigr){\rm d}\mu \notag \\
 &\qquad\quad{}+ \sum_{r \in \{0, \infty\}} \Res_{\mu_r = 0} \Tr_{\ti{2}} \bigl( \iota_{\mu_r} \iota_{\lda_s} r_{\ti{12}}(\lda, \mu) W_{r \ti{2}}(\mu_r)\bigr){\rm d}\mu \notag\\
 &\qquad{}= \sum_{m=1}^\infty W_\infty^{(m)} \lda_\infty^m, \qquad \text{when} \quad s = \infty.
\end{align*}
So, we find that
 \begin{equation}\label{eq:rplusu}
 {\bigl\langle \iota_\mu \iota_\lda r_{\ti{12}}, W_{\ti{2}} \bigr\rangle}_{\ti{2}} = W .
\end{equation}
To evaluate \smash{${\bigl\langle \iota_\mu \iota_\lda r_{\ti{12}}, Z_{\ti{2}} \bigr\rangle}_{\ti{2}}$}, we note that for an arbitrary equivariant rational function ${f \in \mathcal{F}_{D^\prime}^\Gamma}$ and equivariant rational one-form \smash{$g \in \Omega_{D^\prime}^\Gamma$}, we have the following relation:
\begin{align}\label{eq:respropallpoles}
 \Res_{\lda_r = 0} \Tr(g(\lda)f(\lda)){\rm d}\lda &= \Res_{\lda_r = 0} \Tr\bigl(\sgm^k (g(\lda))\sgm^k (f(\lda))\bigr){\rm d}\lda \notag \\
 &= \Res_{\lda_r = 0} \Tr\bigl(\omg ^k g(\omg^k \lda)f(\omg^k \lda)\bigr){\rm d}\lda \notag \\
 &= \Res_{\lda_{r,k} = 0} \Tr(g(\lda)f(\lda)){\rm d}\lda
\end{align}
for all $k \in \{0, \dots, T-1\}$. This allows us to rewrite each entry of the tuple \smash{${\bigl\langle \iota_\mu \iota_\lda r_{\ti{12}}, Z_{\ti{2}} \bigr\rangle}_{\ti{2}}$} as
\begin{align*}
 &T \sum_{r=1}^N \Res_{\mu_r = 0} \Tr_{\ti{2}} \bigl( \iota_{\mu_r} \iota_{\lda_s} r_{\ti{12}}(\lda, \mu) Z_{r \ti{2}}(\mu_r)\bigr){\rm d}\mu \notag \\
 &\qquad\quad{}+ \sum_{r \in \{0, \infty\}} \Res_{\mu_r = 0} \Tr_{\ti{2}} \bigl( \iota_{\mu_r} \iota_{\lda_s} r_{\ti{12}}(\lda, \mu) Z_{r \ti{2}}(\mu_r)\bigr){\rm d}\mu \notag\\
 &\qquad{}= \sum_{r=1}^N \sum_{k=0}^{T-1} \Res_{\mu_{r, k} = 0} \Tr_{\ti{2}} \bigl( \iota_{\mu_r} \iota_{\lda_s} r_{\ti{12}}(\lda, \mu) Z_{r \ti{2}}(\mu_r)\bigr){\rm d}\mu \notag\\
 &\qquad\quad{}+ \sum_{r \in \{0, \infty\}} \Res_{\mu_r = 0} \Tr_{\ti{2}} \bigl( \iota_{\mu_r} \iota_{\lda_s} r_{\ti{12}}(\lda, \mu) Z_{r \ti{2}}(\mu_r)\bigr){\rm d}\mu ,\qquad s \in S ,
\end{align*}
which is a sum over all the residues of a meromorphic one-form on $\mathbb{C}P^1$. Hence, we deduce
\begin{equation}\label{eq:rplusv}
 {\bigl\langle \iota_\mu \iota_\lda r_{\ti{12}}, Z_{\ti{2}} \bigr\rangle}_{\ti{2}} = 0 .
\end{equation}
From \eqref{eq:rplusu} and \eqref{eq:rplusv}, we then have
 \begin{equation*}
 {\bigl\langle \iota_\mu \iota_\lda r_{\ti{12}}, X_{\ti{2}} \bigr\rangle}_{\ti{2}}=W = P_+(X) = R_+(X) .
 \end{equation*}
Since \smash{$X \in \Lg_D^{(0)}$} was arbitrary, this establishes the first relation in \eqref{eq:projrel}. Similarly, using the expansion \eqref{eq:expmur}, we find \smash{${\bigl\langle \iota_\lda \iota_\mu r_{\ti{12}}, W_{\ti{2}} \bigr\rangle}_{\ti{2}} = 0$} and \smash{${\bigl\langle \iota_\lda \iota_\mu r_{\ti{12}}, Z_{\ti{2}} \bigr\rangle}_{\ti{2}} =-Z$} from which we conclude that
\begin{equation*}
 {\bigl\langle \iota_\lda \iota_\mu r_{\ti{12}}, X_{\ti{2}} \bigr\rangle}_{\ti{2}}=-Z = -P_-(X) = R_-(X) .
\end{equation*}
Again, since \smash{$X \in \Lg_D^{(0)}$} was arbitrary, this establishes the second relation in \eqref{eq:projrel}.
\end{proof}

It follows from Proposition \ref{prop:projrel} that the kernel of the linear map $R = P_+ - P_-$ and that of the identity map $\text{Id} = P_+ + P_-$ (with respect to the bilinear pairing of Proposition \ref{prop:pairing}) are
\begin{equation*}
 (\iota_\mu \iota_\lda + \iota_\lda \iota_\mu) r_{\ti{12}}(\lda, \mu)
 \qquad \text{and} \qquad (\iota_\mu \iota_\lda - \iota_\lda \iota_\mu),
 r_{\ti{12}}(\lda, \mu),
\end{equation*}
respectively.

Recall that in Proposition \ref{prop:gddecomp} we identified the Lax matrix \eqref{Lax_matrix_Gaudin} of the cyclotomic Gaudin model, or rather its image under the embedding $\iota_\lambda$ in \eqref{iota map def}, as living in the subspace ${V_- = \iota_\lda \Omega_{D^\prime}^\Gamma}$ of our model \smash{$V = \Lg_D^{(1)}$} for the dual space $\mathbb{g}^\ast$.
The next proposition establishes that this sub\-spa\-ce~$V_- \subset V$ is the orthogonal complement of \smash{$\iota_\lda \mathcal{F}_{D^\prime}^\Gamma \subset \Lg_D^{(0)}$}, i.e., $\mathbb{g}_- \subset \mathbb{g}$, with respect to the~non\-degenerate bilinear pairing of Proposition \ref{prop:pairing}. This completes the proof of the claim that the decomposition $V = V_+ \dotplus V_-$ obtained in Proposition \ref{prop:gddecomp} is the one induced from the corresponding decomposition of the Lie algebra $\mathbb{g} = \mathbb{g}_+ \dotplus \mathbb{g}_-$.

\begin{Proposition}\label{prop:pairequivsp}
An element $Y \in \Lg_D^{(1)}$ lies in $\iota_\lda \Omega_{D^\prime}^\Gamma$ if and only if $\langle Y, X \rangle = 0$ for all $X \in \iota_\lda \mathcal{F}_{D^\prime}^\Gamma$.
Moreover, \smash{$Y \!\in \Lg_D^{(1)}$} lies in \smash{$\Lg_{D+}^{(1)}$} if and only if \smash{$\langle Y, X \rangle = 0$} for all \smash{$X\! \in \Lg_{D+}^{(0)}$}.
\end{Proposition}

\begin{proof}
This is a particular case of the $\Gamma$-equivariant strong residue theorem \cite[Appendix A]{VY1}. We recall the proof here in the present setting for completeness.
 We start by proving that if $Y \in \iota_\lda \Omega_{D^\prime}^\Gamma$, then $\langle Y, X \rangle = 0$ for all $X \in \iota_\lda \mathcal{F}_{D^\prime}^\Gamma$. Let us pick elements $f \in \mathcal{F}_{D^\prime}^\Gamma$ and $g \in \Omega_{D^\prime}^\Gamma$, and let $X = \iota_\lda f$ and $Y = \iota_\lda g$. Since $g(\lambda) f(\lambda) {\rm d}\lambda$ is a meromorphic one-form on $\CP$, by the residue theorem we have
 \begin{equation*}%\label{eq:sumresoneform}
 \sum_{r=1}^N \sum_{k=0}^{T-1} \Res_{\lda_{r, k} = 0} \Tr\bigl(\iota_{\lda_{r, k}}g(\lda)\iota_{\lda_{r, k}}f(\lda) \bigr){\rm d}\lda + \sum_{r \in \{0, \infty\}} \Res_{\lda_r = 0} \Tr\bigl(\iota_{\lda_r}g(\lda)\iota_{\lda_r}f(\lda) \bigr){\rm d}\lda = 0 .
 \end{equation*}
Using the relation \eqref{eq:respropallpoles} valid for any $f \in \mathcal{F}_{D^\prime}^\Gamma$ and $g \in \Omega_{D^\prime}^\Gamma$, we may rewrite the first term on the left-hand side as $T$ times a sum over the residues at the points $\zeta_r$ for $r \in S$, so that
\begin{equation}\label{eq:orthspaceresone}
T \sum_{r=1}^N \Res_{\lda_r = 0} \Tr\bigl( \iota_{\lda_r} g(\lda)\iota_{\lda_r} f(\lda)\bigr){\rm d}\lda + \sum_{r \in \{0, \infty\}} \Res_{\lda_r = 0} \Tr\bigl( \iota_{\lda_r} g(\lda)\iota_{\lda_r} f(\lda)\bigr){\rm d}\lda = 0 .
\end{equation}
By definition of the bilinear pairing from Proposition \ref{prop:pairing}, we thus have $\langle Y, X \rangle = 0$, as desired.

Let us now prove the converse. Namely, let \smash{$Y \in \Lg_D^{(1)}$} be arbitrary and suppose that $\langle Y, X\rangle = 0$ for all $X \in \iota_\lda \mathcal{F}_{D^\prime}^\Gamma$. We must show that, in fact, $Y(\lda) \in \iota_\lda \Omega_{D^\prime}^\Gamma$. It follows from Proposition~\ref{prop:gddecomp} that \smash{$Y = (Y_0, Y_1, \dots, Y_N, Y_\infty) \in \Lg_D^{(1)}$} has a unique decomposition as a direct sum of the tuples \smash{$ \P_+(Y) = (\P_{0+}(Y_0), \P_{1+}(Y_1), \dots, \P_{N+}(Y_N), \P_{\infty+}(Y_\infty)) \in \Lg_{D+}^{(1)} $} and $ \P_-(Y) = \iota_\lda g$ for some ${g \in \Omega_{D^\prime}^\Gamma}$. Now, let $X = \iota_\lda f$ for some $f \in \mathcal{F}_{D^\prime}^\Gamma$. Since $\langle Y, X \rangle = 0$, writing this out explicitly means
\begin{align*}
 &T \sum_{r=1}^N \Res_{\lda_r = 0} \Tr\bigl( \P_{r+}(Y_r(\lda_r)) \iota_{\lda_r}f(\lda) \bigr){\rm d}\lda + \sum_{r \in \{0, \infty\}} \Res_{\lda_r = 0} \Tr\bigl( \P_{r+}(Y_r(\lda_r))\iota_{\lda_r}f(\lda)\bigr){\rm d}\lda \notag \\
 &\qquad{}+T \sum_{r=1}^N \Res_{\lda_r = 0} \Tr\bigl(\iota_{\lda_r}g(\lda)\iota_{\lda_r}f(\lda)\bigr){\rm d}\lda + \sum_{r \in \{0, \infty\}} \Res_{\lda_r = 0} \Tr\bigl(\iota_{\lda_r}g(\lda)\iota_{\lda_r}f(\lda)\bigr){\rm d}\lda = 0 .
\end{align*}
From \eqref{eq:orthspaceresone}, we have that the last two terms on the left-hand side vanish on their own. Therefore, we get
\begin{gather*}
 T \sum_{r=1}^N \Res_{\lda_r = 0} \Tr\bigl( \P_{r+}(Y_r(\lda_r))\iota_{\lda_r}f(\lda)\bigr){\rm d}\lda + \!\!\sum_{r \in \{0, \infty\}}\!\! \Res_{\lda_r = 0} \Tr\bigl( \P_{r+}(Y_r(\lda_r))\iota_{\lda_r}f(\lda)\bigr){\rm d}\lda = 0 .
\end{gather*}
One can then show, along the same lines as the argument in the proof of Proposition \ref{prop:pairing}, that if $\mathcal P_+(Y) \neq 0$ then by picking a suitable $f \in \mathcal{F}_{D^\prime}^\Gamma$ adapted to \smash{$\mathcal{P}_+(Y) \in \Lg_{D+}^{(1)}$}, one can ensure that the expression on the left-hand side is non-zero, which is a contradiction. Therefore, we conclude that $\mathcal P_+(Y) = 0$, and hence $Y \in \iota_\lambda \Omega_{D^\prime}^\Gamma$, as required.
The proof of the ``moreover'' part is completely analogous.
\end{proof}

The upshot of Proposition \ref{prop:pairequivsp} is that we are now exactly in the setting recalled in points~1 and~2 of Section~\ref{review_Lag_multi}. In particular, the coadjoint representation $\Ad^{R\ast} \colon \mathbb{G}_R \times V \to V$ of the group~$\mathbb{G}_R$ on our model \smash{$V = \Lg_D^{(1)}$} for the dual space $\mathbb{g}^\ast$ is given by formula \eqref{AdR}. This will be used below in Lemma \ref{identification_orbits} to give a concise description of the desired coadjoint orbit for the Lax matrix of the cyclotomic Gaudin model.

Before stating the lemma, we first give an explicit description of the infinite-dimensional Lie~group \smash{$\mathbb{G}_+\! = \LG_{D+}^{(0)}$} associated with the Lie algebra \smash{$\mathbb{g}_+\! = \Lg_{D+}^{(0)}$}. Its elements are of the form
\begin{equation} \label{varphi+ def}
 \varphi_+=( \varphi_{0+}, \varphi_{1+}, \dots, \varphi_{N+}, \varphi_{\infty+} ),
\end{equation}
where $\varphi_{r+}(\lda_r)$ is a Taylor series in the local parameter $\lda_r$ with values in $G$, the matrix Lie group associated with the Lie algebra $\g$,
\begin{align}
 & \varphi_{r+}(\lda_r)= \sum_{n=0}^\infty \phi_r^{(n)}\lda_r^n, \qquad r \neq \infty ,\nonumber\\
 & \varphi_{\infty+}(\lda_\infty)= \1 + \sum_{n=1}^\infty \phi_\infty^{(n)}\lda_\infty^n .\label{eq:groupelement}
\end{align}
The Lie group \smash{$\mathbb{G}_+ = \LG_{D+}^{(0)}$} has a natural action on \smash{$V = \Lg_{D}^{(1)}$} given by conjugation. This defines the coadjoint representation $\Ad^\ast \colon \mathbb{G}_+ \times V \to V$ from Section~\ref{review_Lag_multi} explicitly as
\begin{equation*}
(\varphi_+,Y) \longmapsto \Ad^\ast_{\varphi_+} Y=\bigl( \varphi_{0+}Y_0\varphi_{0+}^{-1}, \varphi_{1+}Y_1\varphi_{1+}^{-1}, \dots, \varphi_{N+}Y_N\varphi_{N+}^{-1}, \varphi_{\infty+}Y_\infty\varphi_{\infty+}^{-1} \bigr) .
\end{equation*}
We are now in a position to give an explicit description of the $\mathbb{G}_R$-coadjoint orbit where the Lax matrix of the cyclotomic Gaudin model lives and which will act as our phase space.
\begin{Lemma}\label{identification_orbits}
The orbit of the coadjoint action of $\mathbb{G}_R = \!\bigl(\LG_D^{(0)}\bigr)_R$ on an element ${\iota_\lda \Lambda \in V_- = \iota_\lda \Omega_{D^\prime}^\Gamma}$ has the explicit form
 \begin{equation*}
 \O_\Lambda = \bigl\{ \P_-\bigl(\Ad_{\varphi_+}^\ast \iota_\lda \Lambda \bigr) \mid \varphi_+ \in \LG_{D+}^{(0)} \bigr\} .
 \end{equation*}
\end{Lemma}

\begin{proof}
By definition, the coadjoint orbit of $\bigl(\LG_D^{(0)}\bigr)_R$ on any $\iota_\lda \Lambda \in \iota_\lda \Omega_{D^\prime}^\Gamma$ is given by
\begin{equation*}
\O_\Lambda =\bigl\{ \Ad_{\varphi}^{R\ast} \iota_\lda \Lambda \mid \varphi \in \bigl(\LG_D^{(0)}\bigr)_R \bigr\} .
\end{equation*}
Using the explicit form \eqref{AdR} for the coadjoint action of $\mathbb{G}_R$, we have
\begin{equation*}%\label{eq:rcoadorbdef}
\Ad^{R\ast}_\varphi \iota_\lda \Lambda =\mathcal{P}_-\bigl(\Ad^\ast_{\varphi_+} \iota_\lda \Lambda\bigr)+\mathcal{P}_+\bigl(\Ad^\ast_{\varphi_-}\iota_\lda \Lambda\bigr) = \mathcal{P}_-\bigl(\Ad^\ast_{\varphi_+} \iota_\lda \Lambda\bigr),
\end{equation*}
where the last equality follows from the fact that $\iota_\lambda \Lambda \in V_- = \iota_\lambda \Omega_{D^\prime}^\Gamma$, so that $\Ad^\ast_{\varphi_-}\iota_\lda \Lambda \in V_-$ also and hence $\mathcal{P}_+\bigl(\Ad^\ast_{\varphi_-}\iota_\lda \Lambda\bigr) = 0$.
\end{proof}

It will be useful in practice to express the action of the projector $\P_-$ on $Y \in \Lg_{D}^{(1)}$ as
\begin{equation}
\label{useful}
 \P_-(Y) = \iota_\lda \circ \pi_\lda^{(1)} (Y) .
\end{equation}
In the remaining sections, we will put this setup to use to describe the Lax matrix of the cyclotomic Gaudin model \eqref{Lax_matrix_Gaudin} as a point in a coadjoint orbit $\O_\Lambda$ for some suitable $\Lambda \in \Omega_{D^\prime}^\Gamma$ and then introduce the ingredients from point 3 in Section~\ref{review_Lag_multi} to derive the associated Lax equations.

\subsection{Lax matrix}\label{sec:cgmlaxm}
Our algebraic setup covers the case of the cyclotomic Gaudin model with arbitrary multiplicities. However, as mentioned at the start of this section, from now on we will restrict to the case with simple poles at all $\omg^k \zeta_r$, $k \in \{0, \dots, T-1\}$, $r \in \{1, \dots, N\}$, and double poles at the origin and at infinity, since this is the setting required for our examples in Section~\ref{sec:realisations}. The discussion in the remaining sections is easily generalised to the case of arbitrary multiplicities.

To describe a coadjoint orbit $\O_\Lambda \in \iota_\lda \Omega_{D^\prime}^\Gamma$ from Lemma \ref{identification_orbits} corresponding to the Lax matrix~\eqref{Lax_matrix_Gaudin} of the cyclotomic Gaudin model, we fix a non-dynamical element $\Lambda \in \Omega_{D^\prime}^\Gamma$ with the same pole structure as \eqref{Lax_matrix_Gaudin}, namely we introduce
\begin{equation*}
 \Lambda(\lda) = \frac{\Lambda_0^{(0)}}{\lambda} + \frac{\Lambda_0^{(1)}}{\lambda^2} + \frac{1}{T} \sum_{r=1}^{N} \sum_{k=0}^{T-1} \frac{\sgm^k \Lambda_r}{\lda - \omg^k \zeta_r} + \Lambda_\infty \in \Omega_{D^\prime}^\Gamma .
\end{equation*}
According to Lemma \ref{identification_orbits} and formula \eqref{useful}, the corresponding coadjoint orbit $\O_\Lambda$ then consists of elements of the form
\begin{align}
 & \P_-\Biggl(\varphi_+ \iota_\lda \Biggl( \frac{\Lambda_0^{(0)}}{\lambda} + \frac{\Lambda_0^{(1)}}{\lambda^2} + \frac{1}{T} \sum_{r=1}^{N} \sum_{k=0}^{T-1} \frac{\sgm^k \Lambda_r}{\lda - \omg^k \zeta_r} + \Lambda_\infty \Biggr) \varphi_+^{-1}\Biggr) \notag \\
 &\qquad{}= \iota_\lda \circ \pi_\lda^{(1)} \Biggl( \frac{A_0^{(0)}}{\lda} + \frac{A_0^{(1)}}{\lda^2}, \frac{1}{T} \frac{A_1}{\lda - \zeta_1}, \dots, \frac{1}{T} \frac{A_N}{\lda - \zeta_N}, A_\infty \Biggr) \notag \\
 &\qquad{}= \iota_\lda \Biggl( \frac{A_0^{(0)}}{\lda} + \frac{A_0^{(1)}}{\lda^2} + \frac{1}{T} \sum_{r=1}^{T-1} \sum_{k=0}^{T-1} \frac{ \sgm^k A_r }{\lda - \omg^k \zeta_r} + A_\infty \Biggr)\notag\\
 &\qquad{}\equiv\iota_\lda L(\lda) ,\label{eq:cglaxembed}
\end{align}
where, recalling the definitions \eqref{varphi+ def} and \eqref{eq:groupelement}, we have set
\begin{align}
 &A_0^{(0)} = \phi_0^{(0)} \Lambda_0^{(0)} \phi_0^{(0)\, -1} + \bigl[ \phi_0^{(1)} \phi_0^{(0)\, -1}, \phi_0^{(0)} \Lambda_0^{(1)} \phi_0^{(0)\, -1} \bigr] , \nonumber\\
 &A_0^{(1)} = \phi_0^{(0)} \Lambda_0^{(1)} \phi_0^{(0)\, -1} , \nonumber\\
 &A_r = \phi_r \Lambda_r \phi_r^{-1}, \qquad r \in \{1, \dots, N \} , \nonumber\\
 &A_\infty = \Lambda_\infty ,\label{eq:cgcoadorb}
\end{align}
with $\phi_r^{(0)}$ denoted by $\phi_r$, for $r \in \{1, \dots, N \}$, to simplify notations. Notice that we have the~relations \smash{$\sgm\phi_{0}^{(0)} = \phi_{0}^{(0)}$} and \smash{$\sgm\phi_{0}^{(1)} = \omg\phi_{0}^{(1)}$} for the field elements which ensure that \smash{$A_{0}^{(0)} \in \g^{(0)}$} and \smash{$A_{0}^{(1)} \in \g^{(-1)}$}. This gives us the desired parameterisation of the cyclotomic Gaudin Lax matrix
\begin{equation}\label{eq:cglm}
 L(\lda) = \frac{A_0^{(0)}}{\lda} + \frac{A_0^{(1)}}{\lda^2} + \frac{1}{T} \sum_{r=1}^{N} \sum_{k=0}^{T-1} \frac{ \sgm^k A_r }{\lda - \omg^k \zeta_r} + A_\infty
\end{equation}
viewed as an element of the coadjoint orbit $\O_\Lambda$.

\subsection{Lax equations}\label{sec:cgmlaxeq}
To derive Lax equations for the Lax matrix \eqref{eq:cglm} of the cyclotomic Gaudin model, let us return to the hierarchy of Lax equations \eqref{hierarchy_Lax} induced by the family of Hamiltonians in involution with respect to $\{\, ,\,\}_R$. In our current setup, $\{\, ,\,\}_R$ is the Lie--Poisson bracket on \smash{$\Lg_D^{(1)}$} associated with the linear map $R = P_+ - P_-$. Invariant functions on \smash{$\Lg_D^{(1)}$} take the form
\begin{equation}\label{eq:invfunc}
 H_{p, r} \colon \ Y \in \Lg_D^{(1)} \longmapsto \Res_{\lda_r = 0} \bigl(\ell_{p, r}(\lda_r) \Tr\bigl( Y_r(\lda_r)^{p+1} \bigr) \bigr){\rm d}\lda, \qquad p \geq 1, \qquad r \in S ,
\end{equation}
where $\ell_{p, r}(\lda_r) \in \mathbb{C}((\lda_r))$ is a collection of Laurent polynomials for $p \geq 1$ and $r \in S$. It follows from Proposition \ref{prop:pairing} that for these functions to be non-trivial, the Laurent polynomials $\ell_{p, r}(\lda_r)$ for $r \in \{0, \infty\}$ should be chosen such that \smash{$\ell_{p, r}(\lda_r) Y_r(\lda_r)^p \in \Lg_r^{\Gamma, 0}$}, while $\ell_{p, r}(\lda_r)$ for $r \in \{1, \dots, N \}$ can be any Laurent polynomials. Let us then choose
\begin{gather}
 \ell_{p, r}(\lambda_r) = \iota_{\lda_r} \frac{\lda^p}{p+1} \qquad \text{for}\quad r \in \{0, \infty\}, \nonumber\\
 \ell_{p, r}(\lambda_r) = \iota_{\lda_r} \frac{T \lda^p}{p+1} \qquad \text{for}\quad r \in \{1, \dots, N \} .\label{ell p r choice}
\end{gather}
The restriction of the functions $H_{p, r}$ to $\iota_\lda L$ are Hamiltonians (in involution) of the model and generate the elementary (pairwise-commuting) flows $\partial_{t_p^r}$. However, by virtue of the choice \eqref{ell p r choice} we made, it follows from \eqref{eq:respropallpoles} that $\sum_{r \in S} H_{p,r}(\iota_\lda L) = 0$ by the residue theorem, for each $p \geq 1$.

In what follows, it will thus be sufficient to focus on the Hamiltonians $H_{p, r}$ for $r \neq \infty$ and look at the associated equations of motion they produce through \eqref{hierarchy_Lax}.
Using $R_\pm = \frac{1}{2}(R \pm \text{id})$, these equations can be rewritten as
\begin{equation}\label{eq:laxeqdialg}
 \partial_{t_p^r} \iota_\lda L = \bigl[ R_\pm \nabla H_{p, r}(\iota_\lda L), \iota_\lda L \bigr] .
\end{equation}
As it suffices to calculate only one of the two expressions $R_+ \nabla H_{p, r}(\iota_\lda L)$ and $R_- \nabla H_{p, r}(\iota_\lda L)$, let us compute the latter. The gradient of $H_{p, r}(\iota_\lda L)$ is an element of \smash{\raisebox{-0.3pt}{$\Lg_D^{(0)}$}} and satisfies
\begin{equation} \label{eq:gradham}
 \lim_{\epsilon \rightarrow 0} \frac{H_{p, r}(\iota_\lda L + \epsilon \eta) - H_{p, r}(\iota_\lda L)}{\epsilon} = \langle \eta, \nabla H_{p, r} (\iota_\lda L) \rangle
\end{equation}
for all \smash{$\eta \in \Lg_D^{(1)}$}. Using Proposition \ref{prop:gddecomp}, we may decompose this gradient as follows:
\begin{equation*}
 \nabla H_{p, r} (\iota_\lda L) = N_r^{(p)} + \iota_\lda h_r^{(p)} , \qquad N_r^{(p)} \in \Lg_{D+}^{(0)} , \qquad h_r^{(p)} \in \mathcal{F}_{D^\prime}^\Gamma ,
\end{equation*}
and rewrite \eqref{eq:laxeqdialg} as
\begin{equation*}
 \partial_{t_p^r} \iota_\lda L = [ R_- \nabla H_{p, r}(\iota_\lda L), \iota_\lda L ] = - [ P_-( \nabla H_{p, r} (\iota_\lda L) ), \iota_\lda L ] = - \bigl[\iota_\lda h_r^{(p)}, \iota_\lda L \bigr] .
\end{equation*}
Since $\iota_\lda$ is an embedding which also clearly commutes with the Lie brackets $[\cdot, \cdot] \colon \mathcal{F}_{D^\prime}^\Gamma \times \Omega_{D^\prime}^\Gamma \to \Omega_{D^\prime}^\Gamma$ and \smash{$[\cdot, \cdot] \colon \Lg_D^{(0)} \times \Lg_D^{(1)} \to \Lg_D^{(1)}$}, the above equation implies
\begin{equation}\label{eq:laxeqdialgone}
 \partial_{t_p^r} L = - \bigl[ h_r^{(p)}, L \bigr] .
\end{equation}
In order to calculate \smash{$h_r^{(p)} \in \mathcal{F}_{D^\prime}^\Gamma$}, by virtue of Proposition \ref{prop:pairequivsp}, it is sufficient to restrict $\eta$ in \eqref{eq:gradham} to live in \smash{$\Lg_{D+}^{(1)}$}. We then have
\begin{equation*}
 \bigl\langle \eta, N_r^{(p)} \bigr\rangle = 0\qquad \text{for all}\quad \eta \in \Lg_{D+}^{(1)} ,\quad N_r^{(p)} \in \Lg_{D+}^{(0)} ,
\end{equation*}
and the right-hand side of \eqref{eq:gradham} becomes
\begin{equation*}
 T \sum_{s=1}^N \Res_{\lda_s = 0} \Tr\bigl(\eta_s(\lambda_s) \iota_{\lda_s} h_r^{(p)}\bigr) + \sum_{s \in \{0, \infty\}} \Res_{\lda_s = 0} \Tr\bigl(\eta_s(\lambda_s) \iota_{\lda_s} h_r^{(p)}\bigr) ,
\end{equation*}
while for the left-hand side we find
\begin{equation*}
 \lim_{\epsilon \rightarrow 0} \frac{H_{p, r}(\iota_\lda L + \epsilon \eta) - H_{p, r}(\iota_\lda L)}{\epsilon} = (p+1) \Res_{\lda_r = 0} \Tr\bigl(\ell_{p, r}(\lambda_r)\eta_r(\lambda_r) \iota_{\lda_r} L^p\bigr) , \qquad r \in S ,
\end{equation*}
for any $\eta_s(\lambda_s) \in \Lg_{s+}$, $s \in \{1, \dots, N\}$, and \smash{$\eta_s(\lambda_s) \in \Lg_{s+}^{\Gamma, 1}$}, $s \in \{0, \infty\}$. By definition of $\ell_{p,r}(\lambda_r)$ in~\eqref{ell p r choice}, this implies
\begin{equation*}
 \bigl(\iota_{\lda_s} h_r^{(p)}\bigr)_- = \begin{dcases}
 0 &\text{for}\ s \neq r,\\
 \bigl( \iota_{\lda_r} \lda^p \iota_{\lda_r} L^p \bigr)_- &\text{for}\ s = r \end{dcases}
\end{equation*}
for all $r \in S$, where $X_-$ denotes the principal part of a Laurent series $X$. Plugging the equivariant functions $h_r^{(p)}$ obtained from the above conditions into \eqref{eq:laxeqdialgone} gives a hierarchy of Lax equations corresponding to our choice of invariant functions $H_{p, r}$ in \eqref{eq:invfunc}.

Explicitly, for $p = 1$, we get the Lax equations
\begin{equation}\label{eq:laxeqfirstflow}
 \partial_{t_1^r} L = - \bigl[ h_r^{(1)}, L \bigr] \qquad \text{with} \quad h_r^{(1)} =
 \begin{dcases}
 \dfrac{A_0^{(1)}}{\lda} &\text{for}\ r = 0,\\
 \dfrac{1}{T} \sum_{k=0}^{T-1} \dfrac{\omg^k \zeta_r \sgm^k A_r}{\lda - \omg^k \zeta_r} & \text{for}\ r \in \{1, \dots, N \} .
 \end{dcases}
\end{equation}

Finally, we turn to the primary goal of this paper which is to give a Lagrangian multiform for the cyclotomic Gaudin model that will provide a variational description of the hierarchy of Lax equations in \eqref{eq:laxeqdialgone}. Having just described the cyclotomic Gaudin model within the Lie dialgebra framework, we now have all the necessary ingredients to achieve this goal: the non-dynamical~element $\Lambda \in \Omega_{D^\prime}^\Gamma$ that fixes the phase space, the field element \smash{$\varphi_+ \in \LG_{D+}^{(0)}$} containing~the~dynamical degrees of freedom, the linear map $R$ that equips the phase space with the required Poisson structure, and the invariant functions $H_{p, r}$ which induce non-trivial equations of motion with respect to this Poisson structure.

\subsection{Lagrangian multiform for the cyclotomic Gaudin hierarchy}\label{sec:cgmultiform}
We can now define a Lagrangian $1$-form on the coadjoint orbit ${\cal O}_\Lambda$ of the cyclotomic Gaudin model as
\begin{equation*}%\label{eq:cgmultiform}
 \Lag = \sum_{p=1}^{M} \sum_{r \in S \setminus \{\infty\}} \Lag_{p, r} \, \mathrm{d}t_p^r,
\end{equation*}
using the expression \eqref{our_Lag} for the Lagrangian coefficients. It will be useful to recall the notations in \eqref{eq:cgcoadorb} associated with the parameterisation of the cyclotomic Gaudin Lax matrix. The~elementary times $t_k$ that appear in \eqref{our_Lag} are now naturally labelled by a pair of indices~${p \geq 1}$ and~${r \in S}$, namely we now have elementary times $t_p^r$, associated with the corresponding Hamiltonians~\eqref{eq:invfunc}. Explicitly, we have the following.

\begin{Theorem}\label{th:cgmultiform}
The Lagrangian coefficients of the cyclotomic Gaudin multiform $\Lag$ take the~form%
\begin{equation}\label{eq:cglagcoeff}
 \Lag_{p, r} = \sum_{s = 1}^N \Tr \bigl( A_s \partial_{t_p^r} \phi_s \phi_s^{-1} \bigr) + \Tr \bigl( A_0^{(0)} \partial_{t_p^r} \phi_0^{(0)} \phi_0^{(0)\, -1} \bigr) - H_{p, r}(\iota_\lda L) ,
\end{equation}
with the potential part
\begin{equation}\label{eq:cgpot}
 H_{p, r}(\iota_\lda L) = \Res_{\lda_r = 0} \bigl( \ell_{p, r}(\lda_r) \Tr\bigl(\iota_{\lda_r}L^{p+1}\bigr) \bigr)d\lda, \qquad r \in \{0, 1, \dots, N\} ,
\end{equation}
where $\ell_{p, r}(\lda_r)$ are the Laurent polynomials
\begin{equation*}
 \ell_{p, 0} = \iota_{\lda_0} \frac{\lda^p}{p+1} \qquad \text{and} \qquad \ell_{p, r} = \iota_{\lda_r} \frac{T \lda^p}{p+1} \qquad \text{for}\quad r \in \{1, \dots, N \} .
\end{equation*}
The Lagrangian $1$-form $\Lag$ satisfies the corner equations \eqref{simple_multitime_EL2}--\eqref{simple_multitime_EL3} of the multi-time Euler--Lagrange equations, while the standard Euler--Lagrange equations for $\Lag_{p, r}$ give the hierarchy of Lax equations in \eqref{eq:laxeqdialgone}. Further, on solutions of \eqref{eq:laxeqdialgone}, we have the closure relation
\begin{equation*}
 \partial_{t_q^s}\Lag_{p, r}-\partial_{t_p^r}\Lag_{q, s} = 0
\end{equation*}
for all possible combinations of $(p, r)$ and $(q, s)$ in $\ZZ_{\geq 1} \times S$.
\end{Theorem}
\begin{proof}
Let us start by reinterpreting the formula in \eqref{our_Lag} in the present context of the cyclotomic Gaudin model. First, note that on the coadjoint orbit $\mathcal{O}_\Lambda$, where the Lagrangian $1$-form $\Lag$ lives, the role of the Lax matrix is played by the image of $L$ in \smash{$\Lg_{D}^{(1)}$} under the embedding $\iota_\lda$, given by \eqref{eq:cglaxembed}. Next, the bilinear pairing used to define the kinetic part is the~one constructed in Proposition \ref{prop:pairing}. Furthermore, in the Adler--Kostant--Symes scheme where (locally) $\mathbb{G}_R \simeq \mathbb{G}_+ \times \mathbb{G}_-$, we explicitly have \smash{$\partial_{t_p^r} \varphi \cdot_R \varphi^{-1} = \partial_{t_p^r} \varphi_+ \varphi_+^{-1} + \partial_{t_p^r} \varphi_- \varphi_-^{-1}$}, and since \smash{$\iota_\lambda L \in V_- = \iota_\lambda \Omega_{D^\prime}^\Gamma$}, it follows from Proposition \ref{prop:pairequivsp} that only the \smash{$\partial_{t_p^r} \varphi_+ \varphi_+^{-1}$} piece contributes to the kinetic term. Finally, the potential part is simply the restriction of the invariant functions~$H_{p, r}$ in \eqref{eq:invfunc} to $\iota_\lda L$. Therefore, in the present setup, the Lagrangian coefficients in~\eqref{our_Lag} can be expressed as
\begin{align}
 \Lag_{p, r} &{}= \bigl\langle \iota_\lda L, \partial_{t_p^r} \varphi_+ \varphi_+^{-1} \bigr\rangle - H_{p, r}(\iota_\lda L) \notag \\
 &{}= \bigl\langle \iota_\lda \Lambda, \varphi_+^{-1} \partial_{t_p^r} \varphi_+ \bigr\rangle - H_{p, r}(\iota_\lda L), \qquad r \in \{0, 1, \dots, N\},\label{eq:cglagcoeffinit}
\end{align}
where in the second step we used the fact \eqref{eq:cglaxembed} that $\iota_\lambda L = \mathcal P_-\bigl( \varphi_+ (\iota_\lambda \Lambda) \varphi_+^{-1} \bigr)$. The kinetic term can be written out more explicitly in terms of the component fields \eqref{eq:groupelement} as
\begin{align*}
 \bigl\langle \iota_\lda \Lambda, \varphi_+^{-1} \partial_{t_p^r} \varphi_+ \bigr\rangle &{}= T \sum_{s=1}^N \Res_{\lda_s = 0} \Tr \bigl( \iota_{\lda_s} \Lambda \varphi_{s+}(\lda_s)^{-1} \partial_{t_p^r} \varphi_{s+}(\lda_s) \bigr){\rm d}\lda \notag \\
 &\quad {}+ \sum_{s \in \{0, \infty\} } \Res_{\lda_s = 0} \Tr \bigl( \iota_{\lda_s} \Lambda \varphi_{s+}(\lda_s)^{-1} \partial_{t_p^r} \varphi_{s+}(\lda_s) \bigr){\rm d}\lda \notag \\
 &{}= \sum_{s = 1}^N \Tr \bigl( A_s \partial_{t_p^r} \phi_s \phi_s^{-1} \bigr) + \Tr \bigl( A_0^{(0)} \partial_{t_p^r} \phi_0^{(0)} \phi_0^{(0)\, -1} \bigr) \notag \\
 &\quad {}+ \partial_{t_p^r} \Tr \bigl( \phi_0^{(1)} \phi_0^{(0)\, -1} A_0^{(1)} \bigr) + \frac{1}{2} \partial_{t_p^r} \Tr \bigl(A_\infty \bigl(\phi_\infty^{(1)}\bigr)^2\bigr) .
\end{align*}
The last two terms associated to the poles at the origin and at infinity are total derivatives and do not contribute to the multi-time Euler--Lagrange equations since dropping them amounts to~changing the Lagrangian $1$-form $\Lag$ by $\mathrm{d} \Tr\bigl( \phi_0^{(1)} \phi_0^{(0)\, -1} A_0^{(1)} + \frac{1}{2} A_\infty \bigl(\phi_\infty^{(1)}\bigr)^2 \bigr)$ which is a~total horizontal differential. Discarding these two terms, we are left with the required expression for the Lagrangian coefficients.

Since the Lagrangian coefficients $\Lag_{p, r}$ in \eqref{eq:cglagcoeffinit} are of the form \eqref{our_Lag}, it follows directly from Theorem \ref{Th_multi_EL} that the Lagrangian $1$-form $\Lag$ satisfies all the required conditions and the closure relation.
\end{proof}

To close this section, let us present the explicit expressions for the first set of Lagrangian coefficients $\Lag_{1, r}$ with $r \in \{0, 1, \dots, N \}$. The kinetic terms are obtained by simply substituting $p=1$ in the kinetic part in \eqref{eq:cglagcoeff}, while the potential terms defined by \eqref{eq:cgpot} read
\begin{gather*}
 H_{1, 0} = \frac{1}{2} \Tr \bigl( A_0^{(0)\, 2} \bigr) - \sum_{r=1}^N \frac{\Tr \bigl( A_0^{(1)} A_r \bigr)}{\zeta_r} + \Tr \bigl( A_0^{(1)}A_\infty \bigr),\\
 H_{1, r} = \Tr \bigl( A_0^{(0)}A_r \bigr) + \frac{\Tr\bigl( A_0^{(1)}A_r \bigr)}{\zeta_r} + \frac{1}{2T} \sum_{k=0}^{T-1} \Tr\bigl( A_r \sgm^k A_r \bigr) \\
 \hphantom{H_{1, r} =}{}
 + \frac{1}{T} \sum_{s \neq r} \sum_{k=0}^{T-1} \frac{\Tr\bigl( A_r \sgm^k A_s \bigr) \zeta_r}{\zeta_r - \omg^k \zeta_s} + \Tr(A_r A_\infty) \zeta_r, \qquad r \in \{1, \dots, N\} .
\end{gather*}
Upon varying $\Lag_{1, r}$ with respect to $\phi_s$, $s = 1, \dots, N$, $\phi_0^{(0)}$, and $\phi_0^{(1)}$, it can be checked that the associated Euler--Lagrange equations correspond to the set of Lax equations for $p=1$ in \eqref{eq:laxeqfirstflow}, as it should be from Theorem \ref{th:cgmultiform}.

\section{Realisations of the cyclotomic Gaudin model}\label{sec:realisations}

In this section, we study two different realisations of the cyclotomic Gaudin model -- the periodic Toda chain and the discrete self-trapping (DST) model -- with the objective of describing their corresponding hierarchies variationally. We will then go on to show how our framework allows for a straightforward coupling of integrable hierarchies, by using the approach devised in~\mbox{\cite[Section~7]{CSV}} to couple together hierarchies of integrable field theories.

In what follows, we will work with the cyclotomic Gaudin model associated with the Lie algebra $\g \coloneqq \mathfrak{gl}_T(\mathbb{C})$ and the automorphism $\sgm \in \text{Aut}\, \g$ defined by $\sgm(E_{ij}) = \omg^{j-i}E_{ij}$, for every $i, j = 1, \dots, T$. Here $\omg$ is a primitive $T$-th root of unity, and by $E_{ij}$, $i, j = 1, \dots, T$, we denote the standard basis of $\mathfrak{gl}_T(\mathbb{C})$ taking the indices $i$ and $j$ modulo $T$ by convention. The eigenspaces of $\sgm$ defined by \eqref{eq:sgmeigenspace} are then given by $\g^{(n)} = \text{span}\{E_{i, i+n}\}_{i=1}^T$.

Let us also note the following useful identity that we shall frequently make use of
\begin{equation}\label{eq:omgidentity}
 \frac{z_1^{T-1-[l]} z_2^{[l]}}{z_1^T - z_2^T} = \frac{1}{T} \sum_{k=0}^{T-1} \frac{\omg^{-kl}}{z_1 - \omg^k z_2}
\end{equation}
for any $z_1, z_2 \in \mathbb{C}$ and $l \in \mathbb{Z}$, where $[l] \in \{ 0, \dots, T-1\}$ is such that $l = [l] \mod T$.

\subsection{Periodic Toda chain}\label{sec:toda}

The periodic Toda chain \cite{T} describes a system of particles connected by ``exponential springs'' together with a periodic boundary condition, and has been extensively studied in the Hamiltonian formalism. See, for instance, \cite{F} for the widely used Flaschka change of coordinates, and~\cite{AM} for a proof of its integrability.

{\bf Lax matrix.} We will work with the Lax matrix
\begin{equation}
\label{Lax_Toda}
 L_{\text{Toda}}(\lda) =
 \begin{pmatrix}
		 p_1 \lda^{-1} & 1 & 0 & & \dots & & a_T \lda^{-2}\\
		a_1 \lda^{-2} & p_2 \lda^{-1} & 1 & & \dots & & 0 &\\[1ex]
		\vdots & & \ddots & & & &\vdots\\
		0 & & a_{i-1} \lda^{-2} & p_i \lda^{-1} & 1 & & 0\\
 \vdots & & & & \ddots & &\vdots\\
 0 & & \dots & & a_{T-2} \lda^{-2} & p_{T-1} \lda^{-1} & 1\\
		1 & & \dots & & 0 & a_{T-1} \lda^{-2} & p_T \lda^{-1}
	\end{pmatrix},
\end{equation}
where $a_i = {\rm e}^{q_i - q_{i+1}}$, and $q_i$, $p_i$ are the canonical coordinates satisfying the canonical Poisson bracket relations $\{p_i, q_j\} = \delta_{ij}$ for $i, j = 1, \dots, T$. We also have the periodic boundary conditions $(p_0, q_0) = (p_T, q_T)$ and $(p_{T+1}, q_{T+1}) = (p_1, q_1)$.

The standard Lax matrix for the periodic Toda chain (see \cite[Section~6]{BBT}, for instance)
\begin{equation*}
 \widetilde{L}_{\text{Toda}}(\lda) =
 \begin{pmatrix}
		 p_1 & a_1^{1/2} & 0 & & \dots & & a_T^{1/2} \lda^{-1}\\
		a_1^{1/2} & p_2 & a_2^{1/2} & & \dots & & 0 &\\[1ex]
		\vdots & & \ddots & & & &\vdots\\
		0 & & a_{i-1}^{1/2} & p_i & a_i^{1/2} & & 0\\
 \vdots & & & & \ddots & &\vdots\\
 0 & & \dots & & a_{T-2}^{1/2} & p_{T-1} & a_{T-1}^{1/2}\\
		a_T^{1/2} \lda & & \dots & & 0 & a_{T-1}^{1/2} & p_T \\
	\end{pmatrix}
\end{equation*}
is related to $L_{\text{Toda}}(\lda)$ by conjugation by the diagonal matrix $\mathcal{Q}\! =\! {\rm diag}\bigl({\rm e}^{-q_1/2}\lda^{-1}, \dots, {\rm e}^{-q_T/2}\lda^{-T}\bigr)$ and multiplication by an overall factor of $\lda^{-1}$, together with a change of $\lda$-dependence, as follows:%
\begin{equation}\label{eq:todagaugetransform}
 L_{\text{Toda}}(\lda) = \lda^{-1} \mathcal{Q} \widetilde{L}_{\text{Toda}}\bigl(\lda^T\bigr) \mathcal{Q}^{-1} .
\end{equation}

The Poisson bracket of the Lax matrix $\widetilde{L}_{\text{Toda}}(\lda)$ can be written as
\begin{equation*}
 \bigl\{ \widetilde{L}_{\text{Toda}\, \ti{1}}(\lda), \widetilde{L}_{\text{Toda}\, \ti{2}}(\mu) \bigr\} = \bigl[ \widetilde{r}_{\ti{12}}(\lda, \mu), \widetilde{L}_{\text{Toda}\, \ti{1}}(\lda) + \widetilde{L}_{\text{Toda}\, \ti{2}}(\mu) \bigr],
\end{equation*}
where $\widetilde{r}_{\ti{12}}(\lda, \mu)$ is the skew-symmetric $r$-matrix
\begin{equation*}
 \widetilde{r}_{\ti{12}}(\lda, \mu) = \frac{1}{2} \frac{\mu + \lda}{\mu - \lda} \sum_{i = 1}^T E_{ii} \otimes E_{ii} + \frac{1}{\mu - \lda} \biggl( \mu \sum_{i < j} + \lda \sum_{i > j} \biggr) E_{ij} \otimes E_{ji} .
\end{equation*}
Under the gauge transformation \eqref{eq:todagaugetransform}, on using the identity \eqref{eq:omgidentity}, we find that the Lax matrix $L_{\text{Toda}}(\lda)$ satisfies the Poisson bracket
\begin{equation*}
 \bigl\{L_{\text{Toda}\, \ti{1}}(\lda), L_{\text{Toda}\, \ti{2}}(\mu) \bigr\} = \bigl[r_{\ti{12}}(\lda, \mu), L_{\text{Toda}\, \ti{1}}(\lda) \bigr] - \bigl[r_{\ti{21}}(\mu, \lda), L_{\text{Toda}\, \ti{2}}(\mu)\bigr],
\end{equation*}
where $r_{\ti{12}}(\lambda, \mu)$ is the non-skew-symmetric cyclotomic $r$-matrix
\begin{equation}\label{eq:nss-rmateij}
 r_{\ti{12}}(\lda, \mu) = \frac{1}{T} \sum_{k=0}^{T-1} \frac{\omg^{k(j-i)}}{\mu - \omg^{-k}\lda} E_{ij} \otimes E_{ji} \,
\end{equation}
that we have been working with. This explains our choice of Lax matrix \eqref{Lax_Toda} as opposed to the more traditional one. As just proved, it satisfies the Poisson algebra of the cyclotomic $\mathfrak{gl}_T(\mathbb{C})$-Gaudin model and therefore allows us to obtain the periodic Toda chain as a certain realisation of that model. Specifically, the Lax matrix $L_{\text{Toda}}(\lda)$ can be seen as a realisation of the cyclotomic Gaudin Lax matrix with double poles at the origin and at infinity, that is,
\begin{equation}\label{eq:todalax}
 L_{\text{Toda}}(\lda) = \frac{J_0^{(0)}}{\lda} + \frac{J_0^{(1)}}{\lda^2} + J_\infty,
\end{equation}
where
\begin{equation*}%\label{eq:todacoeff}
 J_0^{(0)} = \sum_{i=1}^T p_i E_{ii} , \qquad J_0^{(1)} = \sum_{i=1}^T {\rm e}^{q_i - q_{i+1}} E_{i+1, i} , \qquad J_\infty = \sum_{i=1}^T E_{i, i+1} .
\end{equation*}

{\bf Orbit realisation.} Set $D = \{0, \infty\}$, and choose the non-dynamical element
\begin{equation*}
 \Lambda_{\text{Toda}}(\lda)= \frac{\Lambda_0^{(1)}}{\lambda^2} + \Lambda_\infty \in \Omega_{D^\prime}^\Gamma,
\end{equation*}
where
\begin{equation*}
\Lambda_0^{(1)} = \sum_{i=1}^T E_{i+1, i} \in \g^{(-1)} , \qquad \Lambda_\infty = \sum_{i=1}^T E_{i, i+1} \in \g^{(1)} .
\end{equation*}
The group elements $\varphi_+ =\left( \varphi_{0+}, \varphi_{\infty+} \right)$, defined by \eqref{eq:groupelement}, contain the dynamical degrees of freedom. From \eqref{eq:cgcoadorb}, we know that the components of $L_{\text{Toda}}(\lda)$ can now be expressed as
\begin{equation}\label{eq:todacoadorb}
 J_0^{(0)} = \bigl[\phi_0^{(1)} \phi_0^{(0)\, -1}, J_0^{(1)}\bigr] , \qquad J_0^{(1)} = \phi_0^{(0)} \Lambda_0^{(1)} \phi_0^{(0)\, -1} , \qquad J_\infty = \Lambda_\infty .
\end{equation}
This gives us a parametrisation of the Lax matrix $L_{\text{Toda}}(\lda)$ in \eqref{eq:todalax} as an element of the coadjoint orbit ${\cal O}_\Lambda^{\text{Toda}}$. Since $\sgm (\varphi_{0+}(\lda)) = \varphi_{0+}(\omg\lda)$, we have \smash{$\sgm \phi_0^{(n)} = \omg^n \phi_0^{(n)}$}. In particular, \smash{$\phi_0^{(0)}$} and \smash{$\phi_0^{(1)}$} have the form
\begin{equation*}
 \phi_0^{(0)} = \sum_{i=1}^T u_i E_{ii} , \qquad \phi_0^{(1)} = \sum_{i=1}^T v_i E_{i, i+1} .
\end{equation*}
For convenience, define $(u_0, v_0) = (u_T, v_T)$ and $(u_{T+1}, v_{T+1}) = (u_1, v_1)$ to encode the periodic boundary conditions. Then, from \eqref{eq:todacoadorb}, we get
\begin{equation*}
 J_0^{(0)} = \sum_{i=1}^T \biggl(\frac{v_i}{u_i} - \frac{v_{i-1}}{u_{i-1}} \biggr) E_{ii} , \qquad J_0^{(1)} = \sum_{i=1}^T \dfrac{u_{i+1}}{u_i} E_{i+1, i} , \qquad J_\infty = \sum_{i=1}^T E_{i, i+1} .
\end{equation*}
Defining
\begin{equation*}
 p_i = \frac{v_i}{u_i} - \frac{v_{i-1}}{u_{i-1}} \qquad \text{and} \qquad q_i = -\ln{u_i} \qquad \text{for}\quad i = 1, \dots, T ,
\end{equation*}
we now have the desired realisation of the coefficients of the Lax matrix $L_{\text{Toda}}(\lda)$. The coadjoint orbit ${\cal O}_\Lambda^{\text{Toda}}$ where $L_{\text{Toda}}(\lda)$ lives is parameterised by $a_i$, $p_i$, $i = 1, \dots, T$, satisfying $\prod_i^T a_i = 1$ and $\sum_i^T p_i = 0$.

{\bf Lax equations.} Let us now look at the Lax equations associated with the invariant functions on \smash{$\Lg_D^{(1)}$} defined in \eqref{eq:cgpot} for the general case. The only functions we need to consider are
\begin{equation*}
 H_{p, 0} = \dfrac{1}{p+1} \Res_{\lda = 0} \bigl(\lda^p \Tr\bigl(\iota_{\lda_0}L_{\text{Toda}}^{p+1}\bigr)\bigr){\rm d}\lda .
\end{equation*}
The set of functions $H_{p, \infty}$ are not independent: we have $H_{p, \infty} = - H_{p, 0}$ for all $p \geq 1$. The Lax equations for the periodic Toda chain with respect to the elementary times $t_p^0$ are given by
\begin{equation*}
 \partial_{t_p^0} \iota_\lda L_{\text{Toda}} = [R_\pm \nabla H_{p, 0}(\iota_\lda L_{\text{Toda}}), \iota_\lda L_{\text{Toda}}] .
\end{equation*}
For $p=1$, these take the form of the set of Lax equations in \eqref{eq:laxeqfirstflow}:
\begin{equation}\label{eq:todalaxeq}
 \partial_{t_1^0} L_{\text{Toda}} = [M_{1, 0}, L_{\text{Toda}}] \qquad \text{with} \quad M_{1, 0} = - \dfrac{J_0^{(1)}}{\lda} .
\end{equation}
Taking residues on both sides in \eqref{eq:todalaxeq} gives the equations of motion for $p_i$, $q_i$ for $i = 1, \dots, T$. To get the equations of motion for $q_i$ it is most convenient to ``undo'' the dressing and write the corresponding equation as
\begin{equation*}
\bigl[\partial_{t_1^0} \phi_0^{(0)} \phi_0^{(0)\, -1}-J_0^{(0)}, J_0^{(1)}\bigr]=0 .
\end{equation*}
This tells us that the diagonal matrix
\[
\phi_0^{(0)\, -1}\bigl(\partial_{t_1^0} \phi_0^{(0)} \phi_0^{(0)\, -1}-J_0^{(0)}\bigr)\phi_0^{(0)}
\]
 must commute with~$\Lambda_0^{(1)}$, and is therefore equal to $\alpha\1$. Using the freedom \smash{$\phi_0^{(0)}\to\phi_0^{(0)} g$}, where $g$ is~diagonal, to set \smash{$\det\phi_0^{(0)}=1$}, we see that $\alpha=0$ by taking the trace of \[
 \partial_{t_1^0} \phi_0^{(0)} \phi_0^{(0)\, -1}-J_0^{(0)}=\alpha\1.
 \] Thus, we have the desired Toda equations
\begin{align}
 &\partial_{t_1^0} q_i = - p_i , \nonumber\\
&\partial_{t_1^0} p_i = {\rm e}^{q_i - q_{i+1}} - {\rm e}^{q_{i-1} - q_i} .\label{eq:todaeom}
\end{align}

{\bf Lagrangian description.} Using Theorem \ref{th:cgmultiform}, we can now write a Lagrangian multiform for the periodic Toda hierarchy as
\begin{equation*}
 \Lag_{\text{Toda}} = \sum_{p=1}^{M} \Lag_{p, 0}\, \mathrm{d}t_p^0
\end{equation*}
with
\begin{equation*}
 \Lag_{p, 0} = - \sum_{i=1}^T p_i \partial_{t_p^0} q_i - \dfrac{1}{p+1} \Res_{\lda = 0} \bigl(\lda^p \Tr\bigl(L_{\text{Toda}}^{p+1}\bigr)\bigr){\rm d}\lda .
\end{equation*}

For $p=1$, this gives us the Lagrangian coefficients for the periodic Toda chain with respect to the time $t_1^0$:
\begin{equation*}
 \Lag_{1, 0} = - \sum_{i=1}^T p_i \partial_{t_1^0} q_i - \frac{1}{2} \sum_{i=1}^T p_i^2 - \sum_{i=1}^T {\rm e}^{q_i - q_{i+1}} .
\end{equation*}
This is the expected {\it phase space} Lagrangian that our method produces and which gives Hamilton's equations \eqref{eq:todaeom} for periodic Toda. It corresponds to the (tangent bundle) Lagrangian
\[
\Lag_{1, 0} = \frac{1}{2}\sum_{i=1}^T \bigl(\partial_{t_1^0} q_i\bigr)^2 - \sum_{i=1}^T {\rm e}^{q_i - q_{i+1}}.
\]

\subsection{DST model}\label{sec:dst}

The discrete self-trapping (DST) equation was introduced in \cite{ELS} to describe the dynamics of~small molecules, which then led to detailed studies of the DST dimer using different methods. The DST model we cast into our framework here is a generalisation of the dimer case to arbitrary (finite) degrees of freedom. This general case first appeared in \cite{CJK} where its relationship with the periodic Toda chain was also hinted at. Our motivation here being different, we do not delve into this connection between the two theories. The interested reader is referred to \cite{KSS} where this aspect was explored further.

{\bf Lax matrix.} We work here with the following avatar of the Lax matrix of the DST model
\begin{equation*}
 L_{{\rm DST}}(\lda) = \dfrac{1}{\lda}\sum_{i=1}^T c_i E_{ii} + \frac{1}{T} \sum_{i, j = 1}^T \sum_{k=0}^{T-1} \frac{\omg^{k(j-i)}x_i X_j E_{ij}}{\lda - \omg^k \zeta_1} + \sum_{i=1}^T E_{i, i+1},
\end{equation*}
where $c_i$, for $i = 1, \dots, T$, are complex parameters, and $x_i$, $X_i$ are the canonical coordinates satisfying the canonical Poisson bracket relations $\{X_i, x_j\} = \delta_{ij}$ for $i, j = 1, \dots, T$, and the periodic conditions $(X_0, x_0) = (X_T, x_T)$ and $(X_{T+1}, x_{T+1}) = (X_1, x_1)$.

The DST Lax matrix in \cite[equation (3.8)]{KSS} is given as
\begin{equation*}
 \widehat{L}_{{\rm DST}}(\mu) = \sum_{i, j =1}^T \frac{b^{T+i-j} x_i X_j E_{ij}}{\mu - b^T} + \mu E_{T1} + \sum_{i \geq j} b^{i-j}x_i X_j E_{ij} + \sum_{i=1}^T c_i E_{ii} + \sum_{i=1}^{T-1} E_{i, i+1},
\end{equation*}
where $b, c_i \in \mathbb{C}$ are parameters of the model, and $\mu$ is the spectral parameter of the Lax matrix. In the present setup where we realise the DST model as a cyclotomic Gaudin model, the role of the parameter $b$ is played by the location of the pole $\zeta_1$ on $\mathbb{C}P^1 \setminus \{0, \infty\}$.
One obtains the Lax matrix $L_{{\rm DST}}(\lda)$ by conjugating $\widehat{L}_{{\rm DST}}\bigl(\lda^T\bigr)$ by the diagonal matrix $\mathcal{D} = {\rm diag}\bigl(\lda^{-1}, \dots, \lda^{-T}\bigr)$ followed by an overall multiplication by $\lda^{-1}$, together with a change of $\lda$-dependence, that is,{\samepage
\begin{equation}\label{eq:dstgaugetransform}
 L_{{\rm DST}}(\lda) = \lda^{-1} \mathcal{D} \widehat{L}_{{\rm DST}}\bigl(\lda^T\bigr) \mathcal{D}^{-1} ,
\end{equation}
and then using the identity \eqref{eq:omgidentity}.}

The Poisson bracket of the Lax matrix $\widehat{L}_{{\rm DST}}(\lda)$ is given as
\begin{equation*}
 \bigl\{\widehat{L}_{{\rm DST}\, \ti{1}}(\lda), \widehat{L}_{{\rm DST}\, \ti{2}}(\mu) \bigr\} = \bigl[\widehat{r}_{\ti{12}}(\lda, \mu), \widehat{L}_{{\rm DST}\, \ti{1}}(\lda)\bigr] - \bigl[\widehat{r}_{\ti{21}}(\mu, \lda), \widehat{L}_{{\rm DST}\, \ti{2}}(\mu)\bigr],
\end{equation*}
where
\begin{equation*}
 \widehat{r}_{\ti{12}}(\lda, \mu) = \frac{1}{\mu - \lda} \biggl( \mu \sum_{i \leq j} + \lda \sum_{i > j} \biggr) E_{ij} \otimes E_{ji} .
\end{equation*}
Using the identity \eqref{eq:omgidentity} once again, one finds that the gauge transformation \eqref{eq:dstgaugetransform} gives for $L_{{\rm DST}}(\lda)$ the $r$-matrix \eqref{eq:nss-rmateij} associated with the cyclotomic $\mathfrak{gl}_T(\mathbb{C})$-Gaudin model, as we would have anticipated. We then have the Poisson bracket
\begin{equation*}
 \bigl\{L_{{\rm DST}\, \ti{1}}(\lda), L_{{\rm DST}\, \ti{2}}(\mu) \bigr\} = \bigl[r_{\ti{12}}(\lda, \mu), L_{{\rm DST}\, \ti{1}}(\lda) \bigr] - \bigl[r_{\ti{21}}(\mu, \lda), L_{{\rm DST}\, \ti{2}}(\mu)\bigr] .
\end{equation*}
Similar to the case of the periodic Toda chain, the Lax matrix $L_{{\rm DST}}(\lda)$ can be seen as a~realisation of the cyclotomic Gaudin Lax matrix, this time with simple poles at the origin and all~$\omg^k \zeta_1$, $k \in \{0, \dots, T-1\}$, for some $\zeta_1 \in \mathbb{C}^{\times}$, and a double pole at infinity, that is,
\begin{equation}\label{eq:dstlax}
 L_{{\rm DST}}(\lda) = \frac{K_0^{(0)}}{\lda} + \frac{1}{T} \sum_{k=0}^{T-1} \frac{ \sgm^k K_1 }{\lda - \omg^k \zeta_1} + K_\infty,
\end{equation}
where
\begin{equation}\label{eq:dstcoeff}
 K_0^{(0)} = \sum_{i=1}^T c_i E_{ii} , \qquad K_1 = \sum_{i, j = 1}^T x_i X_j E_{ij} , \qquad K_\infty = \sum_{i=1}^T E_{i, i+1} .
\end{equation}

{\bf Orbit realisation.} Set $D = \{0, 1, \infty\}$, and choose the non-dynamical element
\begin{equation*}
 \Lambda_{{\rm DST}}(\lda) = \frac{\Lambda_0^{(0)}}{\lambda} + \frac{1}{T} \sum_{k=0}^{T-1} \frac{\sgm^k \Lambda_1}{\lda - \omg^k \zeta_1} + \Lambda_\infty \in \Omega_{D^\prime}^\Gamma,
\end{equation*}
where
\begin{equation*}
\Lambda_0^{(0)} = \sum_{i=1}^T c_i E_{ii} \in \g^{(0)} , \qquad \Lambda_1 = E_{11} \in \g , \qquad \Lambda_\infty = \sum_{i=1}^T E_{i, i+1} \in \g^{(1)} .
\end{equation*}
The coadjoint action of the group elements $\varphi_+ =\left( \varphi_{0+}, \varphi_{1+}, \varphi_{\infty+} \right)$, defined by \eqref{eq:groupelement}, on $\Lambda_{{\rm DST}}(\lda)$ gives the orbit where $L_{{\rm DST}}(\lda)$ lives. From \eqref{eq:cgcoadorb}, we know that the components of $L_{{\rm DST}}(\lda)$ can then be expressed as
\begin{align}\label{eq:dstcoadorb}
 &K_0^{(0)} = \phi_0^{(0)} \Lambda_0^{(0)} \phi_0^{(0)\, -1} , \qquad K_1 = \phi_1 \Lambda_1 \phi_1^{-1} , \qquad K_\infty = \Lambda_\infty ,
\end{align}
where we have denoted \smash{$\phi_1^{(0)}$} by $\phi_1$ for simplicity.
This gives us a parameterisation of the Lax matrix $L_{{\rm DST}}(\lda)$ in \eqref{eq:dsttodalax} as an element of the coadjoint orbit ${\cal O}_\Lambda^{{\rm DST}}$. We parameterise $\phi_1$ as
\begin{equation*}
 \phi_1 = \sum_{i,j=1}^T s_{ij} E_{ij} ,
\end{equation*}
and denote by $\hat{s}_{ij}$, $i, j = 1, \dots, T$, the entries of its inverse, that is,
\begin{equation*}
 \phi_1^{-1} = \sum_{i,j=1}^T \hat{s}_{ij} E_{ij} .
\end{equation*}
From \eqref{eq:dstcoadorb}, we then have
\begin{equation*}
 K_0^{(0)} = \sum_{i=1}^T c_i E_{ii} , \qquad K_1 = \sum_{i,j=1}^T s_{i1} \hat{s}_{1j} E_{ij} , \qquad K_\infty = \sum_{i=1}^T E_{i, i+1} .
\end{equation*}
Defining
\begin{equation*}
 x_i = s_{i1} \qquad \text{and} \qquad X_i = \hat{s}_{1i} \qquad \text{for}\quad i = 1, \dots, T ,
\end{equation*}
gives us the desired realisation of the Lax matrix $L_{{\rm DST}}(\lda)$ as an element of the coadjoint orbit~${\cal O}_\Lambda^{{\rm DST}}$. Notice that $\Tr K_1 = \sum_{i=1}^T x_i X_i = 1$ is an orbit invariant and can be seen as being generated by the symmetry $x_i \rightarrow a x_i$, $X_i \rightarrow a^{-1} X_i$.

{\bf Lax equations.} Let us choose invariant functions $H_{p, r}$ on \smash{$\Lg_D^{(1)} $} as defined in \eqref{eq:cgpot}, and treat~$H_{p, 0}$ and~$H_{p, 1}$ as the independent functions. The Lax equations for the DST model with respect to the elementary times $t_p^r$ are then given by
\begin{equation*}
 \partial_{t_p^r} \iota_\lda L_{{\rm DST}} = [R_\pm \nabla H_{p, r}(\iota_\lda L_{{\rm DST}}), \iota_\lda L_{{\rm DST}}] .
\end{equation*}
For $p=1$, these take the form of the set of Lax equations in \eqref{eq:laxeqfirstflow}:
\begin{equation*}%\label{eq:dstlaxeq}
 \partial_{t_1^r} L_{{\rm DST}} = [M_{1, r}, L_{{\rm DST}}] \qquad \text{with} \quad M_{1, r} =
 \begin{dcases}
 \quad 0 & \text{for}\ r = 0,\\
 - \dfrac{1}{T} \sum_{k=0}^{T-1} \dfrac{\omg^k \zeta_1 \sgm^k K_1}{\lda - \omg^k \zeta_1} & \text{for}\ r = 1 .
 \end{dcases}
\end{equation*}
The $t_1^0$ equations are trivial. Similar to the case of the periodic Toda chain, the easiest way to get the $t_1^1$ equations of motion for $x_i$ and $X_i$, for $i = 1, \dots, T$, is to undo the dressing. With~\smash{${M=K_0^{(0)}+\zeta_1K_\infty+\frac{1}{T} \sum_{k=1}^{T-1} \sgm^k K_1}$}, we find that \smash{$ \phi_1^{-1}(\partial_{t_1^1}\phi_1 \phi_1^{-1}-M)\phi_1$} must commute with $\Lambda_1$, and hence be block diagonal with a scalar ``block'', say $\rho$, and a $(T-1)\times (T-1)$ block which is irrelevant. Then we find, collecting $x_i$ in the vector ${\bf x}$, and $X_i$ in the vector ${\bf X}$,
\begin{equation*}
\partial_{t_1^1}{\bf x}-M{\bf x}=\rho {\bf x}, \qquad \partial_{t_1^1}{\bf X}^T-{\bf X}^TM=\rho {\bf X}^T .
\end{equation*}
Using the freedom $x_i \rightarrow a x_i$, $X_i \rightarrow a^{-1} X_i$ mentioned above with $a={\rm e}^B$, $\partial_{t_1^1}B=\rho$, we can set~${\rho=0}$. Explicitly, we obtain
\begin{gather}
 \partial_{t_1^0} X_i = 0 , \nonumber\\
 \partial_{t_1^0} x_i = 0 , \nonumber\\
 \partial_{t_1^1} X_i = - c_i X_i - \zeta_1 X_{i-1} - \dfrac{1}{T} \sum\limits_{k=1}^{T-1} \sum\limits_{j=1}^T \omg^{k(j-i)}x_j X_j X_i , \nonumber\\
 \partial_{t_1^1} x_i = c_i x_i + \zeta_1 x_{i+1} + \dfrac{1}{T} \sum\limits_{k=1}^{T-1} \sum\limits_{j=1}^T \omg^{k(j-i)} X_j x_j x_i .\label{eq:dsteom}
\end{gather}

{\bf Lagrangian description.} We can now write a Lagrangian multiform for the DST hierarchy using Theorem \ref{th:cgmultiform} as
\begin{equation*}
 \Lag_{{\rm DST}} = \sum_{p=1}^M \Lag_{p, 0}\, \mathrm{d}t_p^0 + \sum_{p=1}^M \Lag_{p, 1}\, \mathrm{d}t_p^1
\end{equation*}
with
\begin{equation*}%\label{eq:dstlagcoeff}
 \Lag_{p, r} = \sum_{i=1}^T X_i \partial_{t_p^r} x_i - H_{p, r}(\iota_\lda L_{{\rm DST}}),\qquad r \in \{0, 1\} ,
\end{equation*}
where the potential term $H_{p, r}(\iota_\lda L_{{\rm DST}})$ is given by \eqref{eq:cgpot} for $r \in \{0, 1\}$. Notice that we have dropped the kinetic contribution to $\Lag_{p, r}$ from the pole at origin since being a total derivative it will not enter the Euler--Lagrange equations.
For $p=1$, the Lagrangian coefficients explicitly~read
\begin{align*}
 &\Lag_{1, 0} = \sum_{i=1}^T X_i \partial_{t_1^0} x_i - \dfrac{1}{2} \sum_{i=1}^T c_i^2 ,\\
 &\Lag_{1, 1} = \sum_{i=1}^T X_i \partial_{t_1^1} x_i - \dfrac{1}{2T} \sum_{i,j=1}^T \sum_{k=0}^{T-1} \omg^{k(j-i)} x_i x_j X_i X_j - \sum_{i=1}^T c_i x_i X_i - \zeta_1 \sum_{i=1}^T x_{i+1} X_i .
 \end{align*}
It can be checked that the Euler--Lagrange equations obtained from varying $\Lag_{1, 0}$ and $\Lag_{1, 1}$ with respect to $X_i$, $x_i$ are exactly the equations in \eqref{eq:dsteom}.

\subsection{Coupled Toda-DST system}\label{sec:todadst}

Finally, we turn to the task of coupling together the two hierarchies we have described variationally in this section. The Lax matrix of the coupled hierarchy can be expressed as
\begin{equation}\label{eq:dsttodalax}
 L_{\text{Toda-DST}}(\lda) = L_{\text{Toda}}(\lda) + \beta L_{{\rm DST}}(\lda) ,
\end{equation}
where $L_{\text{Toda}}(\lda)$ is the Lax matrix of the periodic Toda chain in \eqref{eq:todalax}, $L_{{\rm DST}}(\lda)$ is the DST Lax matrix \eqref{eq:dstlax}, and the parameter $\beta$ is a real-valued scalar parameter dictating the strength of coupling between the two hierarchies. Naturally, the Lax matrix $L_{\text{Toda-DST}}(\lda)$ can be seen as a~realisation of the cyclotomic Gaudin Lax matrix:
\begin{equation*}
 L_{\text{Toda-DST}}(\lda) = \frac{J_0^{(0)} + \beta K_0^{(0)}}{\lda} + \frac{J_0^{(1)}}{\lda^2} + \frac{\beta}{T} \sum_{k=0}^{T-1} \frac{ \sgm^k K_1 }{\lda - \omg^k \zeta_1} + J_\infty + \beta K_\infty
\end{equation*}
with the $\g$-valued coefficients defined in \eqref{eq:todacoadorb} and \eqref{eq:dstcoeff}.

As the number of finite (non-zero) poles in a cyclotomic Gaudin model is arbitrary, our construction allows us, in principle, to couple together an arbitrary number of copies of the DST model and a copy of the periodic Toda chain. Here we only illustrate it for a single copy each of the periodic Toda chain and the DST model.

{\bf Orbit realisation.} Set $D = \{0, 1, \infty\}$. Since we already have a parameterisation of the Lax matrices $L_{\text{Toda}}(\lda)$ and $L_{{\rm DST}}(\lda)$ as orbit elements, we only need to check that the action of a~generic group element \smash{$\varphi_+ =( \varphi_{0+}, \varphi_{1+}, \varphi_{\infty+}) \in \LG_{D+}^{(0)}$} defined by \eqref{eq:groupelement} on the non-dynamical element
\begin{equation*}
 \Lambda_{\text{Toda-DST}}(\lda) = \Lambda_{\text{Toda}}(\lda) + \beta \Lambda_{{\rm DST}}(\lda) \in \Omega_{D^\prime}^\Gamma,
\end{equation*}
where
\begin{equation*}
 \Lambda_{\text{Toda}}(\lda)= \frac{\Lambda_0^{(1)}}{\lambda^2} + \Lambda_\infty , \qquad \Lambda_{{\rm DST}}(\lda) = \frac{\Lambda_0^{(0)}}{\lambda} + \frac{1}{T} \sum_{k=0}^{T-1} \frac{\sgm^k \Lambda_1}{\lda - \omg^k \zeta_1} + \Lambda_\infty
\end{equation*}
results in the Lax matrix $L_{\text{Toda-DST}}(\lda)$. Indeed, using the result in \eqref{eq:cgcoadorb}, we have that the components of $L_{\text{Toda-DST}}(\lda)$ take the form
\begin{align*}
 &J_0^{(0)} + \beta K_0^{(0)} = \bigl[\phi_0^{(1)} \phi_0^{(0)\, -1}, J_0^{(1)}\bigr] + \beta \phi_0^{(0)} \Lambda_0^{(0)} \phi_0^{(0)\, -1} ,\nonumber\\
 &J_0^{(1)} = \phi_0^{(0)} \Lambda_0^{(1)} \phi_0^{(0)\, -1} ,\nonumber\\
 &\beta K_1 = \beta \phi_1 \Lambda_1 \phi_1^{-1} , \qquad \text{and} \nonumber\\
 &J_\infty + \beta K_\infty = \Lambda_\infty + \beta \Lambda_\infty ,%\label{eq:todadstcoadorb}
\end{align*}
as desired. This gives us a parameterisation of the Lax matrix $L_{\text{Toda-DST}}(\lda)$ in \eqref{eq:dsttodalax} as an element of the coadjoint orbit ${\cal O}_\Lambda^{\text{Toda}} \times {\cal O}_\Lambda^{{\rm DST}}$.

{\bf Lax equations.} As earlier, we will choose invariant functions $H_{p, r}$ on \smash{$\Lg_D^{(1)} $} as defined in~\eqref{eq:cgpot}, for $r \in \{0, 1\}$. The Lax equations for the DST-Toda model with respect to the elementary times~$t_p^r$ are given by
\begin{equation*}
 \partial_{t_p^r} \iota_\lda L_{\text{Toda-DST}} = [R_\pm \nabla H_{p, r}(\iota_\lda L_{\text{Toda-DST}}), \iota_\lda L_{\text{Toda-DST}}] .
\end{equation*}
For $p=1$, these take the form of the set of Lax equations in \eqref{eq:laxeqfirstflow}:
\begin{equation}\label{eq:todadstlaxeq}
 \partial_{t_1^r} L_{\text{Toda-DST}} = [M_{1, r}, L_{\text{Toda-DST}}]
\end{equation}
with
\begin{equation*}
 M_{1, r} =
 \begin{dcases}
 - \dfrac{J_0^{(1)}}{\lda} &\text{for}\ r = 0,\\
 - \dfrac{\beta}{T} \sum_{k=0}^{T-1} \dfrac{\omg^k \zeta_1 \sgm^k K_1}{\lda - \omg^k \zeta_1} & \text{for}\ r = 1 .
 \end{dcases}
\end{equation*}
Taking residues on both sides in \eqref{eq:todadstlaxeq} gives the equations of motion for $p_i$, $q_i$, $X_i$, $x_i$, for ${i = 1, \dots, T}$. Like in the cases of the periodic Toda chain and the DST model, to get the equations of motion for $q_i$, $X_i$, $x_i$, one can ``undo'' the dressing in the corresponding equations. Explicitly, one gets the following equations for $p_i$, $q_i$ for $i = 1, \dots, T$:
\begin{gather}
 \partial_{t_1^0} p_i = (1 + \beta)\bigl({\rm e}^{q_i - q_{i+1}} - {\rm e}^{q_{i-1} - q_i}\bigr) + \dfrac{\beta}{\zeta_1} \bigl({\rm e}^{q_{i-1} - q_i} x_{i-1} X_i - {\rm e}^{q_i - q_{i+1}} x_i X_{i+1}\bigr) , \nonumber\\
 \partial_{t_1^0} q_i = - p_i - \beta c_i , \nonumber\\
 \partial_{t_1^1} p_i = \dfrac{\beta}{\zeta_1} \bigl({\rm e}^{q_i - q_{i+1}} x_i X_{i+1} - {\rm e}^{q_{i-1} - q_i} x_{i-1} X_i\bigr) , \nonumber\\
 \partial_{t_1^1} q_i = - \beta x_i X_i ,\label{eq:todaeombeta}
\end{gather}
and the following equations for $X_i$, $x_i$, for $i = 1, \dots, T$:
\begin{gather}
 \partial_{t_1^0} X_i = \dfrac{1}{\zeta_1} {\rm e}^{q_i - q_{i+1}} X_{i+1} , \nonumber\\
 \partial_{t_1^0} x_i = - \dfrac{1}{\zeta_1} {\rm e}^{q_{i-1} - q_i}x_{i-1} , \nonumber\\
 \partial_{t_1^1} X_i = - p_i X_i - \beta c_i X_i - \dfrac{1}{\zeta_1} {\rm e}^{q_i - q_{i+1}} X_{i+1} - \dfrac{\beta}{T} \sum\limits_{j=1}^T \sum\limits_{k=0}^{T-1} \omg^{k(j-i)} x_j X_j X_i - (1 + \beta) \zeta_1 X_{i-1} ,\nonumber\\
 \partial_{t_1^1} x_i = p_i x_i + \beta c_i x_i + \dfrac{1}{\zeta_1} {\rm e}^{q_{i-1} - q_i}x_{i-1} + \dfrac{\beta}{T} \sum\limits_{j=1}^T \sum\limits_{k=0}^{T-1} \omg^{k(j-i)} X_j x_j x_i + (1 + \beta) \zeta_1 x_{i+1} .\label{eq:dsteombeta}
\end{gather}

Setting $\beta = 0$ in the above equations produces the equations of motion in \eqref{eq:todaeom} for the periodic Toda chain. In the limit $\beta \rightarrow \infty$, these reduce to the equations of motion we obtained for the DST model in \eqref{eq:dsteom}. To see this, note that the Lax matrix $L_{{\rm DST}}$ comes as $\beta L_{{\rm DST}}$ in the coupled system. Therefore, in the limit $\beta \rightarrow \infty$, the time flow $t_p^r$ is rescaled such that $\partial_{t_p^r}$ rescales to $\beta^p \partial_{t_p^r}$.

{\bf Lagrangian description.} Using Theorem \ref{th:cgmultiform}, we can now write a Lagrangian multiform for the Toda-DST hierarchy:
\begin{equation*}
 \Lag_{\text{Toda-DST}} = \sum_{p=1}^M \Lag_{p, 0} \mathrm{d}t_p^0 + \sum_{p=1}^M \Lag_{p, 1} \mathrm{d}t_p^1
\end{equation*}
with
\begin{align*}
\Lag_{p, r} &{}= \Tr \bigl(J_0^{(0)} \partial_{t_p^r} \phi_0^{(0)} \phi_0^{(0)\, -1}\bigr) + \beta \Tr \bigl(K_0^{(0)} \partial_{t_p^r} \phi_0^{(0)} \phi_0^{(0)\, -1}\bigr) \notag \\
 &\quad{} + \beta \Tr \bigl(K_1 \partial_{t_p^r} \phi_1 \phi_1^{-1}\bigr) - H_{p, r}(\iota_\lda L_{\text{Toda-DST}}),\qquad r \in \{0, 1\} ,
\end{align*}
where the potential term $H_{p, r}(\iota_\lda L_{\text{Toda-DST}})$ is given by \eqref{eq:cgpot} for $r \in \{0, 1\}$. Notice that $K_0^{(0)} = \sum_{i=1}^T c_i E_{ii}$. So, the second term on the right-hand side is, in fact, a total derivative and will not enter the Euler--Lagrange equations. So, we will simply drop it. In terms of the canonical coordinates, we then have
\begin{equation}\label{eq:dtlagcoeff}
 \Lag_{p, r} = - \sum_{i=1}^T p_i \partial_{t_p^r} q_i + \beta \sum_{i=1}^T X_i \partial_{t_p^r} x_i - H_{p, r}(\iota_\lda L_{\text{Toda-DST}}),\qquad r \in \{0, 1\} .
\end{equation}
The decoupled limits of the periodic Toda and the DST hierarchies are easily recovered by setting $\beta = 0$ and taking the limit $\beta \rightarrow \infty$ in \eqref{eq:dtlagcoeff}, respectively.

For $p=1$, \eqref{eq:dtlagcoeff} gives the Lagrangian for the coupling of the periodic Toda chain and the~DST model with the potential terms
\begin{gather*}
 H_{1, 0} = \frac{1}{2} \sum_{i=1}^T p_i^2 + \beta \sum_{i=1}^T c_i p_i + \frac{\beta^2}{2} \sum_{i=1}^T c_i^2 + (1 + \beta) \sum_{i=1}^T {\rm e}^{q_i - q_{i+1}} - \frac{\beta}{\zeta_1} \sum_{i=1}^T {\rm e}^{q_i - q_{i+1}} x_i X_{i+1} , \\
 H_{1, 1} = \frac{\beta^2}{2T} \sum_{i,j=1}^T \sum_{k=0}^{T-1} \omg^{k(j-i)} x_i x_j X_i X_j + \beta \sum_{i=1}^T p_i x_i X_i + \beta^2 \sum_{i=1}^T c_i x_i X_i \\
 \hphantom{H_{1, 1} =}{}
 + \frac{\beta}{\zeta_1} \sum_{i=1}^T {\rm e}^{q_i - q_{i+1}} x_i X_{i+1} + \bigl(\beta + \beta^2\bigr) \zeta_1 \sum_{i=1}^T x_{i+1} X_i .
\end{gather*}

The variation $\delta \Lag_{1, r}$ gives the Euler--Lagrange equations for the coupled system with respect to the time flow $t_1^r$, for each $r \in \{0, 1\}$. It can be checked that these are exactly the equations we obtained in \eqref{eq:todaeombeta} and \eqref{eq:dsteombeta}.

\section{Conclusion}\label{sec:conclusion}

From the point of view of Lagrangian multiforms, our main result is the construction of a~Lagrangian multiform for the class of cyclotomic Gaudin models which includes, as a particular example, the emblematic periodic Toda chain. The main technical aspect of this result lies in the pinning down of the necessary algebraic setup (in Section~\ref{dialgebra_Gaudin}) associated with the classical $r$-matrix of the cyclotomic Gaudin model which is both non-skew-symmetric and depends on spectral parameters. This is the first time that a Lagrangian multiform has been constructed for an $r$-matrix with both of these properties, as previous general constructions were either only presented for skew-symmetric ones \cite{CDS, CSV} or for non-skew-symmetric $r$-matrices but without spectral parameter \cite{CDS}.

The non-skew-symmetric $r$-matrix of the cyclotomic Gaudin model arises from a decomposition of a direct sum of (twisted) loop algebras into a pair of complementary subalgebras, which controls the corresponding Lie dialgebra structure entering its Lagrangian multiform. While finishing this work, the preprint \cite{AMS} appeared on the construction of solutions of the general classical Yang--Baxter equation based on more general decompositions of such Lie algebras and to which one can associate generalised notions of Gaudin models. In particular, it appears that the cyclotomic setup considered here is a special case of their construction. It would be interesting to investigate if our Lagrangian multiform construction can be extended to the more general non-skew-symmetric $r$-matrices and associated generalised Gaudin models of \cite{AMS}.

\subsection*{Acknowledgements}

The authors would like to thank the referees for their constructive feedback. A.A.S. is funded by the School of Mathematics EPSRC Doctoral Training Partnership Studentship (Project Reference Number 2704447). B.V.\ gratefully acknowledges the support of the Leverhulme Trust through a Leverhulme Research Project Grant (RPG-2021-154).

\pdfbookmark[1]{References}{ref}
\LastPageEnding

\end{document}